\documentclass[11pt]{article}

\usepackage{palatino}
\usepackage{cite}
\usepackage{amsmath,amssymb,amsfonts}
\usepackage{algorithmic}
\usepackage{graphicx}
\usepackage{textcomp}
\usepackage{booktabs}
\usepackage{soul}
\usepackage{arydshln}
\usepackage{caption}
\usepackage{subcaption}
\usepackage{url}
\usepackage{hyperref}
\usepackage{stfloats}
\usepackage{braket}
\usepackage{booktabs}
\usepackage{bm}
\usepackage[utf8]{inputenc}
\usepackage{mathtools}
\usepackage{graphicx}
\usepackage{algorithm}
\usepackage{cite}
\usepackage{amsmath,amssymb,amsfonts, amsthm, mathrsfs}
\usepackage{algorithmic}
\usepackage{textcomp}
\usepackage{xcolor}
\usepackage{enumitem}

\newtheorem{lemma}{Lemma}
\newtheorem{theorem}{Theorem}

\newtheorem{remark}{Remark}

\newcommand\rk{\normalfont{\mbox{rk}}}

\newcommand{\Tr}{\normalfont{\mbox{Tr}}}

\newcommand\init{\normalfont {\sf init}}

\newcommand\achi{\normalfont{\footnotesize \mbox{achi}}}

\def\BibTeX{{\rm B\kern-.05em{\sc i\kern-.025em b}\kern-.08em
    T\kern-.1667em\lower.7ex\hbox{E}\kern-.125emX}}

\def\bbsmatrix#1{\begin{bsmallmatrix}#1\end{bsmallmatrix}}

\allowdisplaybreaks

\title{On the Capacity of Vector Linear Computation  over a Noiseless Quantum Multiple Access Channel with Entangled Transmitters}
\author{Yuhang Yao, Syed A. Jafar\\
{\small Center for Pervasive Communications and Computing (CPCC)}\\
{\small University of California Irvine, Irvine, CA 92697}\\
{\small \it Email: \{yuhangy5, syed\}@uci.edu}
}
\date{}

\begin{document}
\maketitle
\begin{abstract}
Network function computation is an active topic in network coding, with much recent progress for  linear (over a finite field) computations over broadcast (LCBC) and multiple access (LCMAC) channels. Over a quantum multiple access channel  (QMAC) with quantum-entanglement shared among transmitters, the linear computation  problem (LC-QMAC) is non-trivial even when the channel  is noiseless, because of the challenge of optimally exploiting transmit-side entanglement through distributed coding. Given an arbitrary linear function of data streams defined in a finite field $\mathbb{F}_d$, the LC-QMAC problem seeks the optimal communication cost (minimum number of qudits that need to be sent by the transmitters to the receiver, per computation instance) over a noise-free QMAC, when the independent input data streams originate at the corresponding transmitters, who share quantum entanglement in advance.   As our main result, we fully solve this problem for $K=3$ transmitters ($K\geq 4$ settings remain open). Coding schemes based on the $N$-sum box protocol (along with time-sharing and batch-processing) are shown to be information theoretically optimal in all cases. 
\end{abstract}

\maketitle

\section{Introduction}
With the much-anticipated quantum technologies appearing on the horizon \cite{caleffi_tutorial2}, there is increasing interest in exploring the potential impacts on communication and computation capabilities. In particular, distributed encoding of \emph{classical} information into entangled \emph{quantum} systems over \emph{many-to-one}  communication networks is a cross-cutting theme across a variety of active research areas that include  quantum private information retrieval (QPIR) \cite{allaix2020quantum, QMDSTPIR,Lu_Jafar_QXSTPIR,AytekinIT25,AytekinGC23,AytekinITW24}, quantum metrology and sensing \cite{Lloyd, DQS_QZ,Rubio_2020},  quantum machine learning  \cite{QML1,QML2} and quantum simultaneous message passing \cite{kawachi2021communication,christensen2023private}. By exploiting uniquely quantum phenomena such as  entanglement and superposition, the hybrid classical-quantum (CQ) paradigm promises precision, security, privacy and efficiency guarantees beyond the fundamental limits of purely classical systems. This may be accomplished, for example, by sending the entangled quantum systems to a central receiver that extracts the desired information through a joint measurement. 

In order to understand the fundamental limits of many-to-one CQ systems it is imperative to study the classical information carrying capacity of a quantum multiple access (QMAC) channel. One approach in this direction focuses on the challenges posed by \emph{noisy} quantum channels, both for  \emph{communication} tasks  --- where the receiver's goal is to recover the transmitters' data inputs (messages) \cite{winter_multiple_access_capacity, bennett_shor_capacity, hsieh2008entanglement, yard2008capacity,Shi_Guha_QMAC}, as well as \emph{computation} tasks --- where the receiver only wishes to retrieve a particular  function (e.g., sum) of the inputs \cite{sohail2022unified, sohail2022computing}. Advances in this direction tend to require quantum generalizations of classical random coding arguments, made especially challenging by the superadditivity of quantum capacity \cite{superadditivity} which presents obstacles to single-letterization. Remarkably, even for a  \emph{point-to-point} noisy quantum channel, a  computable closed form capacity expression is not always available.

A different approach,   called the LC-QMAC problem \cite{Yao_Jafar_Sum_MAC, Allaix_N_sum_box, Yao_Jafar_SQEMAC,Ulukus_Guha}, emerged relatively recently out of QPIR literature \cite{QMDSTPIR,Lu_Jafar_QXSTPIR,AytekinIT25,AytekinGC23,AytekinITW24} and focuses exclusively on the utility of transmitter-side\footnote{Prior entanglement with the receiver is not assumed by default in the LC-QMAC, but can be modeled by including a dummy transmitter as in \cite{Yao_Jafar_SQEMAC}.} quantum \emph{entanglement} for linear computation (LC) tasks  under idealized assumptions on the QMAC, e.g., the channels through which the quantum systems are delivered to the receiver may be assumed to be noise-free. The noise-free model ensures that the capacity reflects the fundamental limits of \emph{entanglement} as a resource for computation, rather than those of the underlying noise models and associated countermeasures. Essentially in this case, \emph{the entanglement is the channel}, i.e., quantum entanglement introduces non-classical dependencies between the distributed quantum systems, which collectively constitute a non-trivial channel. Intuitively, the  challenge  is  to  extract as much distributed superdense coding gain \cite{DSC2,DSC4, christensen2023private, Petar2025,Dutta2023} as possible through distributed coding and joint measurements to match the desired computation task at the receiver, thereby \emph{maximizing the efficiency (capacity) of the communication resource (qubits) required for the desired computation}. Idealized channel models  make the problem  more tractable --- optimal coding schemes under this approach are more likely to be non-asymptotic, and the capacity more likely to be found in closed form, thus somewhat transparent and insightful. Indeed, this is the case when the function to be computed is simply a sum of the transmitters' inputs \cite{Yao_Jafar_Sum_MAC}. The LC-QMAC approach  seeks a resource theoretic accounting  analogous to the degrees of freedom (DoF) studies of wireless networks \cite{Jafar_FnT} where the noise is similarly de-emphasized. It is a quantum extension of the classical topic of \emph{network function computation} \cite{Appuswamy3, Huang_Tan_Yang_Guang,  Ramamoorthy_Langberg,  Yao_Jafar_3LCBC, Yao_Jafar_KLCBC, Derya_CBC}, and as such is relevant to applications that seek communication-efficient computation, such as QPIR  \cite{QMDSTPIR,Lu_Jafar_QXSTPIR,AytekinIT25,AytekinGC23,AytekinITW24}.

It is important to note that despite the simplification afforded by idealized (rather than noisy) channel models the LC-QMAC problem remains challenging because of the long recognized \cite{Korner_Marton_sum} increased difficulty of characterizing the capacity for \emph{computation} (rather than \emph{communication}) tasks, as evident from the abundance of open problems in network function computation. The present work falls under the LC-QMAC paradigm. See Fig. \ref{fig:LCQMAC} for an illustration of the LC-QMAC problem considered in this work. A formal description is presented in Section \ref{sec:LCQMAC}.
\begin{figure}[htbp]
\center
\includegraphics[height=0.4\textwidth]{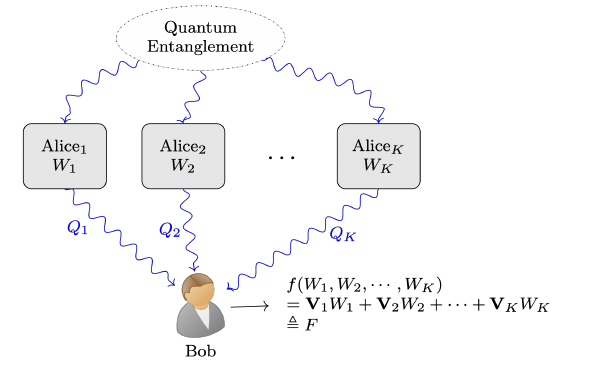}
\caption{LC-QMAC($\mathbb{F}_d, K, {\bm V}_1, {\bm V}_2, \cdots, {\bm V}_K$). $Q_1, Q_2,$ $ \cdots, Q_K$ are entangled quantum systems. Alice$_k$ encodes $W_k$ into $Q_k$, and Bob measures the joint system $Q_1Q_2\cdots Q_K$ to obtain the desired computation $F$.} \label{fig:LCQMAC}
\end{figure}

\subsection{Background: $N$-sum Box  for Linear Computation over a QMAC (LC-QMAC)} \label{sec:background}
\begin{figure*}[t]
\center
\includegraphics[height=0.4\textwidth]{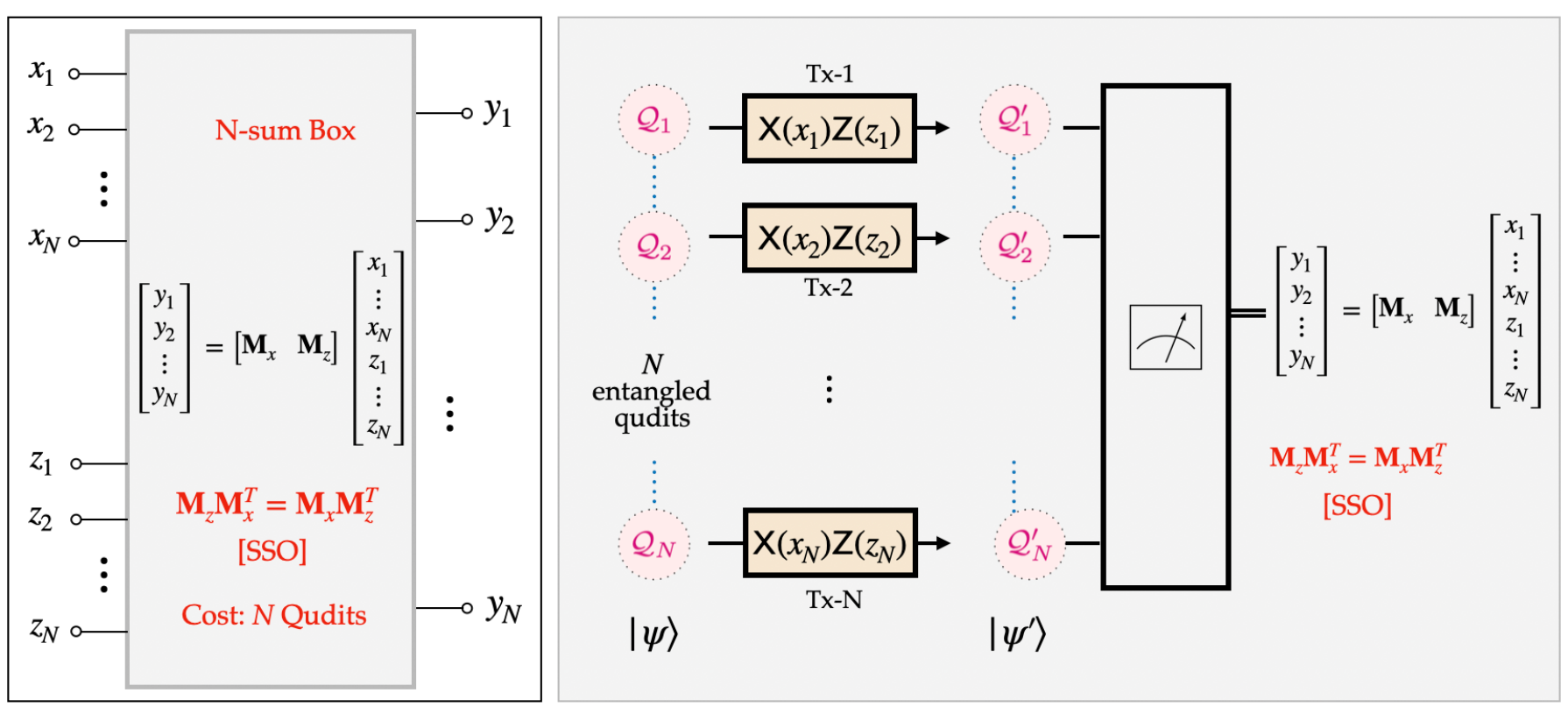}
\caption{The $N$-sum box \cite{Allaix_N_sum_box} is shown on the left as a black-box abstraction of a quantum protocol for classical distributed many-to-one linear computation. The actual quantum protocol is shown on the right (details can be found in \cite{Allaix_N_sum_box}). $N$ qudits are initially prepared in a suitable stabilizer state $\ket{\psi}$ and distributed to $N$  transmitters, the classical inputs $(x_n,z_n)$ are applied by the $n^{th}$ transmitter to manipulate the $n^{th}$ qudit via conditional Pauli $X$ and $Z$ gates, all $N$ qudits are sent to a receiver (so the communication cost of the protocol is $N$ qudits), and a joint measurement at the receiver produces the linear function of the inputs, ${\bm y} = {\bm M}_x {\bm x}+{\bm M}_z {\bm z}$. Given any ${\bm M}_x, {\bm M}_z\in\mathbb{F}_d^{N\times N}$ there exists a stabilizer state $\ket{\psi}$ and a measurement that realizes this computation functionality, provided rank$[{\bm M}_x, {\bm M}_z]=N$ and ${\bm M}_x{\bm M}_z^\top = {\bm M}_z{\bm M}_x^\top$ (strong self-orthogonality).} \label{fig:Nsumbox}
\end{figure*}

As the starting point for this work, consider the $N$-sum box protocol formalized in  \cite{Allaix_N_sum_box}, which specifies a set of $\mathbb{F}_d$ linear functions that can be computed over an ideal (noise-free) $N$-to-$1$ QMAC, with $N$-qudits being transmitted to a central receiver, one each from each of $N$ transmitters who share quantum entanglement in advance but are not otherwise allowed to communicate with each other. Specifically, if the $n^{th}$ transmitter, $n\in[N]$, has classical inputs $(x_n, z_n) \in \mathbb{F}_d^2$ which it encodes into its own qudit by local Pauli $X,Z$ operations, then after receiving $1$ noise-free qudit per transmitter, following the $N$-sum box protocol, the receiver is able to obtain ${\bm y} = {\bm M}_x {\bm x}+{\bm M}_z {\bm z}$, where ${\bm x} = [x_1,\cdots,x_n]^\top, {\bm z} = [z_1,\cdots,z_n]^\top$, and ${\bm M}_x, {\bm M}_z$ are $N\times N$ matrices in $\mathbb{F}_d$ such that rank$[{\bm M}_x, {\bm M}_z]=N$ and ${\bm M}_x{\bm M}_z^\top = {\bm M}_z{\bm M}_x^\top$. The last condition is called the \emph{strong self-orthogonality} (SSO) condition. See Figure \ref{fig:Nsumbox} for an illustration.

The significance of the SSO condition can be briefly summarized as follows. The $N$-sum box protocol is built on the framework of stabilizer codes. The protocol requires that the qudits be measured with respect to a set of commutative observables called the stabilizers, that also determine the initial entangled state. The commutativity of the stabilizers, required for such a protocol, manifests as the SSO condition. The smallest concrete $N$-sum box protocol is a $2$-sum box. Consider matrices ${\bm M}_x =\bbsmatrix{1&1\\0&0}, {\bm M}_z =\bbsmatrix{0&0\\1&-1}$ in $\mathbb{F}_d$. The SSO property is readily verified. It follows that there exists a $2$-sum box protocol where Transmitter $n$, $n\in[2]$, has inputs $(x_n,z_n)\in \mathbb{F}_d^2$, sends $1$ encoded qudit to the receiver, and the receiver jointly measures the $2$ qudits to obtain the outputs ${\bm y}=\bbsmatrix{x_1+x_2\\ z_1-z_2}$.

It is worth mentioning that the $N$-sum box protocol emerged out of the QPIR literature and was formalized in \cite{Allaix_N_sum_box} primarily as a useful abstraction that hides the details of the underlying quantum coding schemes, and thereby makes these quantum coding applications accessible to classical coding and information theorists.

\subsection{Motivating Examples} \label{sec:motivating_examples}
Let us motivate this work with three toy examples. Toy Example 1 illustrates the standard problem formulation in the `forward' (easy) direction, i.e., given SSO matrices one can apply the $N$-sum box protocol to find out what linear functions can be computed and at what communication cost. Toy Example 2 illustrates the problem formulation in the `inverse' (harder) direction, i.e., given a desired linear computation find the \emph{optimal} (maximally efficient) protocol to accomplish it. It is important to note that we mean optimality in a strong information theoretic sense, i.e., not limited to $N$-sum box protocols.  Toy Example 2 presents a relatively simple case of the inverse question that can be solved with existing bounds. Finally, Toy Example 3 shows how existing bounds are \emph{insufficient} to answer the inverse question, thereby motivating the work in this paper.

\subsubsection{Toy Example 1} Given  the matrices ${\bm M}_x =\bbsmatrix{1&1&1\\0&0&0\\0&0&0}, {\bm M}_z =\bbsmatrix{0&0&0\\1&2&0\\1&0&2}$, say over $\mathbb{F}_d, d=3$, it is readily verified that the SSO property is satisfied, giving us an $N$-sum box ($N=3$) with output ${\bm y}=\bbsmatrix{x_1+x_2+x_3\\ z_1+2z_2\\ z_1+2z_3}$. The box can be used for example, in an LC-QMAC setting where we have $3$  transmitters: Alice$_1$, Alice$_2$, Alice$_3$, with prior shared quantum entanglement, who are presented with independent classical input streams $(A,B), (C,D),(E,F)$, respectively, all symbols in $\mathbb{F}_3$, and a receiver (Bob) who wishes to compute, 
\begin{align*}
	f(A,B,C,D,E,F)=\bbsmatrix{A+C+E\\ B+2D\\ B+2F}.
\end{align*}
The total download cost incurred by the $N$-sum box solution in this case is $3$ qudits. In fact, the scheme is  information theoretically optimal in its communication cost because with i.i.d. uniform inputs the entropy $H(f(A,B,C,D,E,F))=3$ dits, and Holevo's bound implies that $3$ dits (in this case meaning $d=3$-ary digits) worth of information cannot be delivered by fewer than $3$ qudits. By the same reasoning, given arbitrary SSO matrices ${\bm M}_x, {\bm M}_z$ we can identify the corresponding linear function that is optimally computed by the $N$-sum box protocol in an LC-QMAC setting. 

\subsubsection{Toy Example 2} Now let us consider an `inverted' situation, i.e., instead of the ${\bm M}_x, {\bm M}_z$ matrices, we are given a desired linear function to be computed over a given QMAC. For example, suppose the three  transmitters, Alice$_1$, Alice$_2$, Alice$_3$, have classical input data streams $(A), (B),(C)$, respectively, all symbols in $\mathbb{F}_3$, and Bob (the receiver) wishes to compute $g(A,B,C)=\bbsmatrix{A+B+C},$ i.e., the sum of the three data-streams.  Since the entropy of $g(A,B,C)$ is at most $1$ dit per instance,  Holevo's bound only indicates that the communication cost is at least $1$ qudit per instance of $g$. One could try to \emph{search} for an $N$-sum box (i.e., SSO matrices ${\bm M}_x,{\bm M}_z$) that can output $g(A,B,C)$ at the total communication cost equal to (or approaching asymptotically with joint coding across many computation instances) $1$ qudit per instance, but such a search would be futile. This is because an information theoretic (min-cut) argument (cf. \cite{Yao_Jafar_Sum_MAC}) shows that no quantum coding scheme can allow Bob to recover $g(A,B,C)$ at a cost less than $1.5$ qudits per computation.\footnote{We will occasionally drop the qualifier `per computation' for the sake of brevity, with the understanding that download costs are always measured per instance of the desired function computation.} The optimal total download cost is indeed $1.5$ qudits in this case, and it is achievable with the $N$-sum box protocol \cite{Yao_Jafar_Sum_MAC} by coding over $L=2$ instances so that $A=(A_1,A_2), B=(B_1,B_2), C=(C_1,C_2)$. In fact
the same $N$-sum box as in the previous example suffices, by setting ${\bm x}=\bbsmatrix{A_1&B_1&C_1}^\top$ and ${\bm z}=\bbsmatrix{A_2&B_2&C_2}^\top$, which produces output $\bbsmatrix{A_1+B_1+C_1\\ A_2+2B_2\\ A_2+2C_2}$. Note that once Bob recovers both $A_2+2B_2$ and $A_2+2C_2$, he can add them and divide the sum by $2$  to  recover $A_2+B_2+C_2$. The inverted problem formulation --- finding a suitable $N$-sum box protocol given the desired computation --- is perhaps more natural. However, the inverted problem is challenging when the desired computation does not directly correspond to an SSO matrix structure, and therefore may need to be minimally expanded (e.g., by breaking $A_2+B_2+C_2$ into $A_2+2B_2$ and $A_2+2C_2$ as in this toy example) into a larger computation that does fit an SSO structure. We note the recent progress in this direction in \cite{Ulukus_Guha} which investigates $N$-sum box based coding schemes with more than $3$ transmitters.

\subsubsection{Toy Example 3}\label{sec:toyex3}
Suppose the $3$ transmitters Alice$_1$, Alice$_2$, Alice$_3$, have classical input streams $(A), (B), (C,D)$, respectively, all symbols in $\mathbb{F}_3$, and Bob wishes to compute the function 
\begin{align*}
	h(A,B,C,D)=\bbsmatrix{A+B+C\\ D}.
\end{align*}
Applying Holevo's bound for this case only shows that the communication cost must be at least $2$ qudits. Min-cut arguments also produce the same bound. However, a search for such an $N$-sum box fails, leading to the question: \emph{Does there always exist an $N$-sum box protocol that achieves the information theoretically minimal download cost per computation given an arbitrary desired linear computation over a QMAC?} More generally, \emph{what is the optimal communication cost per computation instance for an arbitrary desired  linear computation over a QMAC, and how can it be achieved?} For the particular setting of Toy Example 3, it turns out that what is needed is a stronger information theoretic converse bound (see Theorem \ref{thm:new_bounds} in this work), that will show  that the optimal communication cost is at least $2.5$ qudits (per computation).  In fact, if $\Delta_1,\Delta_2,\Delta_3$ represent the number of qudits (per computation instance) sent to Bob from Alice$_1$, Alice$_2$, Alice$_3$, respectively, then the (closure of) set of \emph{all} feasible tuples is characterized as follows (see Theorem \ref{thm:main_K3} in this work).
{\small
\begin{align}
	\mathfrak{D}^* = 
	\left\{
	\begin{bmatrix}
		\Delta_1\\\Delta_2\\\Delta_3
	\end{bmatrix}
	\in \mathbb{R}^3
	\left|
	\begin{array}{l}
		\Delta_1 \geq 1/2\\
		\Delta_2 \geq 1/2\\
		\Delta_3 \geq 1\\
		\Delta_1+\Delta_2+\Delta_3 \geq 5/2
	\end{array}
	\right.
	\right\}.\label{eq:ex3}
\end{align}
}%
See Fig. \ref{fig:ex3} for an illustration of  the optimal region $\mathfrak{D}^*$ of all feasible tuples, as well as an optimal coding scheme that utilizes a $5$-sum box protocol. The claim that $\mathfrak{D}^*$ is information-theoretically optimal (i.e., that there cannot exist any other scheme capable of achieving a better communication cost) requires a matching impossibility result, which follows from the proof of converse of Theorem \ref{thm:new_bounds}, presented in Section \ref{proof:converse}.
\begin{figure*}[htbp]
\center	
\includegraphics[width=0.25\textwidth]{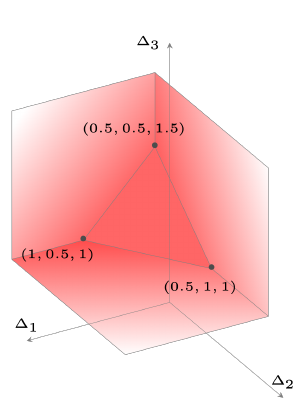}\includegraphics[width=0.75\textwidth]{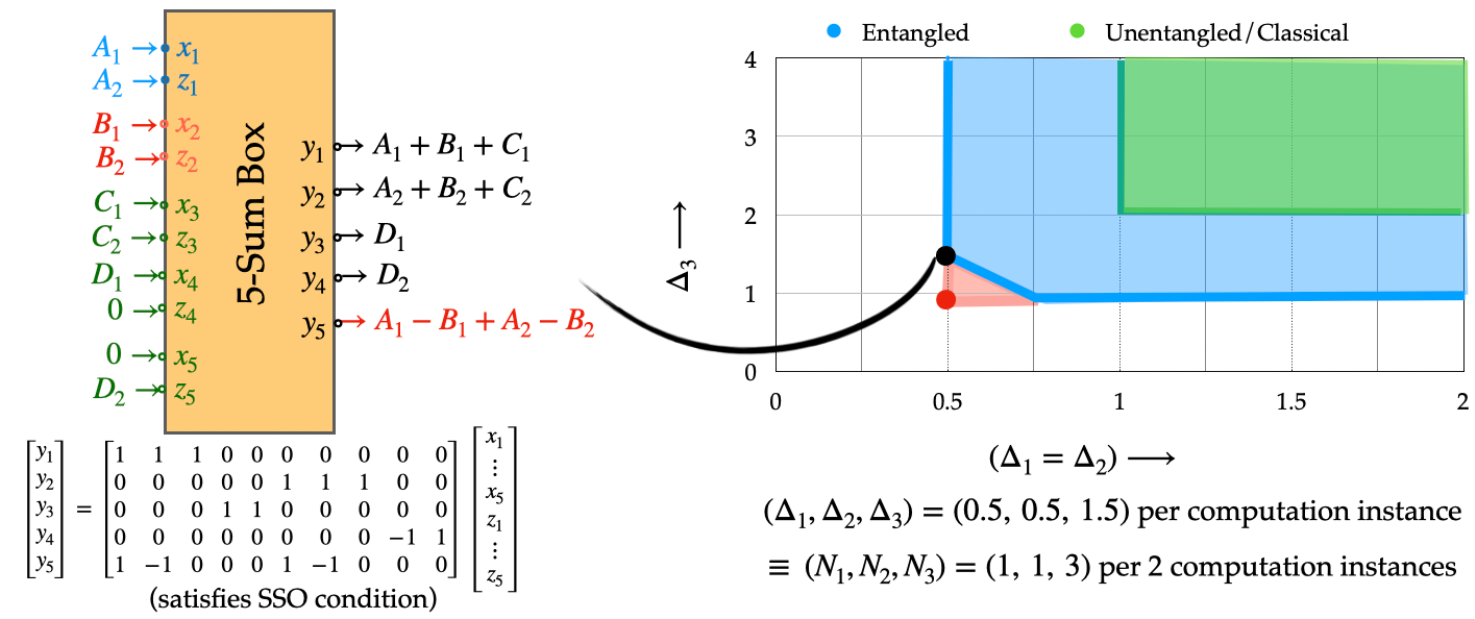}
\caption{$\mathfrak{D}^*$ for Toy Example 3 is shown on the left. A coding scheme over $\mathbb{F}_3$ utilizing a $5$-sum box protocol is shown in the middle, achieving $2$ computations of the desired function $(A+B+C,D)$ at a communication cost of $N_1,N_2,N_3=1,1,3$ qudits from Alice$_1$, Alice$_2$, Alice$_3$, who have  input streams $(A), (B), (C,D)$, respectively. A projection of  $\mathfrak{D}^*$ into $2$ dimensions, by setting $\Delta_1=\Delta_2$, is shown on the right (blue region), along with the  unentangled/classical feasible region (green, contained in blue, obtained directly from classical cut-set bounds). The black dot $(\Delta_1,\Delta_2,\Delta_3)=(0.5,0.5,1.5)$ per computation, corresponds to the scheme illustrated in the middle. The region that is outside the blue region, e.g., the red dot $(\Delta_1,\Delta_2,\Delta_3)=(0.5,0.5,1)$, is not feasible by \emph{any} coding scheme, i.e., not even with other protocols that may not rely on the $N$-sum box, as shown by the information theoretic converse of Theorem \ref{thm:new_bounds}.}\label{fig:ex3}
\end{figure*}

\subsection{Overview of Contribution}
\subsubsection{Key Questions}
To summarize the motivating examples, while the $N$-sum box abstraction specifies what can be computed {\it given} any choice of SSO matrices ${\bm M}_x, {\bm M}_z$, typically we are much more interested in the \emph{inverted} problem formulation. A general network function computation application may require any particular $\mathbb{F}_d$ linear function $f$ of the transmitters' inputs. The desired function need not satisfy any SSO condition. In fact most linear functions $f({\bm x},{\bm z})$ cannot directly be expressed as $f({\bm x},{\bm z}) = {\bm M}_x {\bm x}+{\bm M}_z {\bm z}$ for some ${\bm M}_x, {\bm M}_z$ that satisfy an SSO condition. What typically matters to an application is the {\it cost} of the {\it desired} computation. So the key objective is to find the most efficient protocol, i.e., the information theoretically optimal protocol for any arbitrary given $\mathbb{F}_d$ linear function. Notably, the case where $f$ is simply the sum of inputs has been settled in \cite{Yao_Jafar_Sum_MAC}, and coding schemes based on the $N$-sum box are shown to be capacity achieving in that case. However, in the general case where $f$ can be an arbitrary vector linear function, it is far from obvious what the optimal cost might be for computing $f$ on a QMAC; whether that cost is achievable with an $N$-sum box protocol; if so, then how can it be achieved; and if not, then what else may be needed. In particular, the SSO constraint that limits the scope of $N$-sum box functionality is  quite intriguing. Does it represent a fundamental information theoretic limitation? If so, then how does it translate into entropic constraints? Or is it merely an artifact of the  $N$-sum box protocol that may be circumvented by other, more general constructions? Remarkably, it follows from \cite{Yao_Jafar_Sum_MAC} that the SSO constraint does not pose a limitation for the $K=2$ transmitter setting.\footnote{This is because for linear computations the $2$-sum box allows full cooperation between the two transmitters \cite{Yao_Jafar_Sum_MAC}.} Therefore, the smallest case that is open is the $3$-to-$1$ LC-QMAC setting, which is indeed our main focus in this paper. The main contribution of this work is to answer the aforementioned questions fully for the $K=3$ transmitter setting.  Let us note, however, that the question remains open for $K>3$ transmitters.

\subsubsection{Summary of Results}
Specifically, our main result is a solution to the inverted problem identified above, hence labeled an \emph{inverted $3$-sum box}. Given \emph{any} desired $\mathbb{F}_d$ linear computation $f$ (not limited to scalar linear functions as in \cite{Yao_Jafar_Sum_MAC}) on a $3$-to-$1$ QMAC, the \emph{inverted $3$-sum box} solution provides, 
\begin{enumerate}
\item[-] a region $\mathfrak{D}^*$ of download cost (per computation instance) \emph{tuples} $(\Delta_1, \Delta_2,\Delta_3)$ corresponding to Alice$_1$, Alice$_2$, Alice$_3$, such that each of these tuples is sufficient for the desired computation (note that this is a \emph{region} of tuples, so we are not limited to just the total download cost, or to symmetric download costs),
\item[-] a coding scheme that makes use of only the $N$-sum box protocol and  TQC to achieve the desired computation for any feasible download cost  tuple in $\mathfrak{D}^*$, and
\item[-] an information theoretic converse which shows that for any download cost tuple outside the set $\mathfrak{D}^*$ the function $f$ cannot be computed by \emph{any} coding scheme (not limited to just the $N$-sum box or TQC schemes).
\end{enumerate}
The result establishes the information theoretic optimality of the $N$-sum box protocol for the $K=3$ transmitter LC-QMAC. Interestingly, this is indicative of the information theoretic significance of the SSO constraint, since the achievable schemes that are limited primarily by the SSO constraint, end up being information theoretically optimal.

Last but not the least, since we focus on the $3$ transmitter LC-QMAC, let us recall a somewhat surprising observation from \cite{Yao_Jafar_Sum_MAC}, that $3$-way entanglement is never \emph{necessary} to achieve capacity in the $\Sigma$-QMAC. The $\Sigma$-QMAC is a special case of the LC-QMAC where the desired computation is simply a sum of data-streams, like the setting of Toy Example 2. Recall that coding schemes based on the $N$-sum box are  sufficient for achieving the capacity of the $\Sigma$-QMAC in \cite{Yao_Jafar_Sum_MAC}. In particular, \cite{Yao_Jafar_Sum_MAC} shows that any coding scheme for a $\Sigma$-QMAC that utilizes a $3$-sum box, can be translated into an equally efficient coding scheme that utilizes only $2$-sum boxes, and therefore only $2$-way entanglements. For instance, in Toy Example $2$, we note that $A+B+C$ can  be computed equally efficiently with only $2$-sum boxes by computing $f_1(A,B)=A_1-A_2+B_1, f_2(B,C)=B_2-C_1+C_2, f_3(A,C)=A_2+C_1$, each of which requires only a $2$-sum box, and then recovering the desired computations as $f_1+f_3=A_1+B_1+C_1=g(A_1,B_1,C_1)$ and $f_2+f_3=A_2+B_2+C_2=g(A_2,B_2,C_2)$, for the same total download cost of $1.5$ qudits per computation instance. 

Remarkably, we find that this is no longer the case when the scope of desired computations is expanded from the $\Sigma$-QMAC to the LC-QMAC, i.e., instead of only a sum of inputs, the desired computation can be an arbitrary \emph{vector} linear combination of inputs, as in this paper. Indeed, $3$-way entanglements are \emph{necessary} in general for \emph{vector} linear computations. We establish this non-trivial fact  by providing an information theoretic proof  that $3$-way entanglements between the transmitters are necessary in the $3$-transmitter LC-QMAC setting of Toy Example 1 in order  to achieve the optimal cost of $3$ qudits per computation. 
Specifically, we prove in the Appendix \ref{proof:necessity3} that with only $2$-way entanglements (which allow $2$-sum boxes) the  total download cost for Toy Example $1$ cannot be less than $3.5$ qudits per computation. While the proof is non-trivial, partial intuition can be gained from the observation that $2$-way entanglement at best allows any two transmitters to collaborate, i.e., to send any coded symbols of their joint database at a cost of one qudit/dit.  The proof is then done by bounding the download cost of a classical LC-MAC where each transmitter knows the data streams of a pair of transmitters in the original LC-QMAC. \\

\noindent {\it Notation:} For $n\in \mathbb{N}$, define $[n] \triangleq \{1,2,\cdots, n\}$. For $a < b \in \mathbb{N}$ define $[a:b]=\{a,a+1,\cdots,b\}$. Given a set $\mathcal{S}$, define $A_\mathcal{S}\triangleq\{A_s \mid s\in \mathcal{S}\}$. $\mathbb{F}_d$ denotes the finite field with order $d$ being a power of a prime. For a matrix $M \in \mathbb{F}_d^{a\times b}$, $\rk(M)$ denotes its rank over $\mathbb{F}_d$. $\mathbb{R}$ and $\mathbb{Q}$ denote the set of reals and rationals, respectively. For vectors $u,v$ of the same length, $u\geq v$ is equivalent to $u_i \geq v_i, \forall i$ where $u_i, v_i$ are the $i^{th}$ component of $u$ and $v$, respectively.  Given a tripartite quantum system $ABC$ in the state $\rho$, $H(A)_\rho$ denotes the entropy of $A$ with respect to the state $\rho$. The conditional entropy $H(A\mid B)_\rho$ is defined as $H(AB)_\rho - H(B)_\rho$ and the conditional mutual information is defined as $I(A;B\mid C)_\rho = H(A\mid C)_\rho + H(B\mid C)_\rho - H(AB \mid C)_\rho$. The subscript in the information measures may be omitted for compact notation when the underlying state is obvious from the context. If the state additionally depends on a classical random variable $X$ with distribution $p_X$, and say $\rho$ denotes the joint state of the classical-quantum system, then $H(A\mid X=x)_\rho$ denotes the entropy of $A$ conditioned on the event $X=x$. Similar to classical information measure, we have $H(A\mid X)_\rho = \sum_x p_X(x)H(A\mid X=x)_\rho$.

\section{Problem Formulation}
\subsection{LC-QMAC} \label{sec:LCQMAC}
An LC-QMAC setting (see Fig. \ref{fig:LCQMAC}) is specified by the parameters $(\mathbb{F}_d, K,{\bm V}_1, \cdots, {\bm V}_K)$. $\mathbb{F}_d$ is a finite field of order $d$. $K$ is the number of transmitters (denoted as Alice$_k, k\in[K]$). For $k\in [K]$, ${\bm V}_k$ is an $m\times m_k$ matrix with elements in $\mathbb{F}_d$. Alice$_k, k\in[K]$ has a data stream $W_k$, which takes values in $\mathbb{F}_d^{m_k\times 1}$, and the receiver, Bob, wants to compute an arbitrary $\mathbb{F}_d$ linear function of the data streams, $F = {\bm V}_1W_1+\cdots+{\bm V}_KW_K \in \mathbb{F}_d^{m\times 1}$. Without loss of generality we assume that for all $k\in[K]$,
\begin{enumerate}
	\item $m_k \leq m$;
	\item ${\bm V}_k$ has full column rank.
\end{enumerate}
The desired computation is to be performed multiple times, for successive instances of the data streams. Specifically, for $\ell \in \mathbb{N}$, the realization of the data stream $W_k$ corresponding to the $\ell^{th}$ instance of the computation is denoted as $W_k^{\ell}$. Denote $W_k^{[L]} = [W_k^{1}, W_k^{2}, \cdots, W_k^{L}]$. The $\ell^{th}$ instance of the function to be computed is then identified as $F^{\ell}$ and we have the compact notation $F^{[L]} = [F^1, F^2,\cdots, F^L]$.

\subsection{Coding Schemes for LC-QMAC}
For the LC-QMAC  $(\mathbb{F}_d, K,{\bm V}_1, \cdots, {\bm V}_K)$, a (quantum) coding scheme involves the following elements.
\begin{itemize}
	\item A batch size $L \in \mathbb{N}$,  which represents the number of computation instances to be encoded together by the coding scheme.
	\item A composite quantum system $Q=Q_1Q_2\cdots Q_K$ comprised of $K$ subsystems, with initial state of $Q$ specified by the density matrix $\rho^{\init}$.
	\item A set of encoders represented as quantum channels $\{\mathcal{E}_k^{(w_k)} \colon k\in[K], w_k\in \mathbb{F}_d^{m_k\times L}\}$, such that the output dimension of each $\mathcal{E}_k^{(w_k)}$ is equal to $\delta_k$.
	\item A set of operators $\{\Lambda_y\colon y\in \mathcal{Y} \}$ that specify a POVM.
\end{itemize}
See Fig. \ref{fig:coding_scheme} for an illustration of a quantum coding scheme. The coding scheme is explained as follows. There are three stages, referred to as the preparation stage, the encoding stage, and the decoding stage.
\begin{figure}[htbp]
\center
\includegraphics{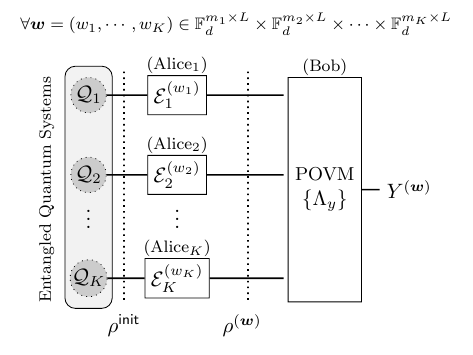}
\caption{A quantum coding scheme for the LC-QMAC. The output measured at the receiver, $Y ^{({\bm w})}$, must be equal to ${\bm V}_1w_1+{\bm V}_2 w_2+\cdots+{\bm V}_Kw_K$, for all realizations of  $(w_1,w_2,\cdots, w_K)$.}
\label{fig:coding_scheme}
\end{figure}

\begin{enumerate}
	\item \textbf{(Preparation stage):} A $K$ partite quantum system $Q_1Q_2\cdots Q_K$ is prepared in the initial state $\rho^{\init}$ and distributed to the Alices such that for all $k\in[K]$, Alice$_k$ has the subsystem $Q_k$.
	\item \textbf{(Encoding stage):} For data realization (over $L$ instances) 
		 {\small \begin{align*}
		 	(W_1^{[L]}, W_2^{[L]},\cdots, W_K^{[L]}) = (w_1,w_2,\cdots, w_K), 
		 \end{align*} }%
		 Alice$_k$ applies $\mathcal{E}_k^{(w_k)}$ to $\mathcal{Q}_k$ for $k\in [K]$. The output state of the composite quantum system is thus determined as,
		 {\small
		\begin{align}
			\rho^{(w_1,\cdots,w_K)} = \mathcal{E}_1^{(w_1)}\otimes \mathcal{E}_2^{(w_2)}\otimes \cdots \otimes \mathcal{E}_K^{(w_K)} (\rho^{\init}).
		\end{align}}%
	\item \textbf{(Decoding stage):} Bob measures $Q_1Q_2\cdots Q_K$ with POVM $\{\Lambda_y\colon y\in \mathcal{Y} \}$ to obtain the output random variable $Y$, such that,
		{\small \begin{align} \label{eq:state}
			\Pr ( Y = y) = \Tr(\rho^{(w_1,\cdots,w_K)} \Lambda_y),~~ \forall y\in \mathcal{Y}.
		\end{align}}%
\end{enumerate}
A feasible coding scheme must satisfy the following correctness condition,
{\small
\begin{align} \label{eq:correctness}
	&\mbox{[Correctness]} && \Pr (Y = F^{[L]}) = 1,
\end{align}
}%
for all  realizations of the data streams $(w_1,\cdots, w_K) \in \mathbb{F}_d^{m_1\times L} \times \cdots \times \mathbb{F}_d^{m_K \times L}$.

\subsection{Download Cost Tuple} 
Given a feasible coding scheme, define
{\small
\begin{align}
	{\bm \Delta} = (\Delta_1,\cdots, \Delta_K) = \left(\frac{\log_d \delta_1}{L}, 
	\cdots, \frac{\log_d \delta_K}{L}\right)
\end{align}
}%
as the normalized download cost tuple (simply referred to as the cost tuple in the rest of the paper) achieved by the coding scheme. Specifically, for $k\in [K]$, $\Delta_k$ measures the average number of qudits downloaded from Alice$_k$, normalized by the number of computation instances $L$. A cost tuple is said to be \emph{achievable} if it is achieved by some feasible coding scheme.

\subsection{Optimal Cost Region}
For an LC-QMAC, the optimal cost region $\mathfrak{D}^*$ is defined as the closure of the set of all achievable cost tuples. Specifically, let $\mathfrak{C}_L$ denote the set of feasible coding schemes with batch size $L$. Let ${\bm \Delta}(\mathcal{C})$ denote the cost tuple achieved by the coding scheme $\mathcal{C}$. Define $\mathfrak{D}_L = \{{\bm \Delta}(\mathcal{C})\colon \mathcal{C} \in \mathfrak{C}_L\}$. Then 
{\small
\begin{align}
	\mathfrak{D}^* \triangleq  \overline{ \bigcup_{L=1}^{\infty} \mathfrak{D}_L},
\end{align}
}%
where $\overline{X}$ denotes the closure of $X$ in $\mathbb{R}^K$.

\section{Preliminaries}
We briefly review some relevant known results.
\subsection{$N$-sum box}
Formally, an $N$-sum box is specified by a finite field $\mathbb{F}_q$, a  matrix ${\bm M} = [{\bm M}_x, {\bm M}_z]$ where ${\bm M}_x$, ${\bm M}_z \in \mathbb{F}_q^{N\times N}$ such that $\rk({\bm M}) = N$ and ${\bm M}_x {\bm M}_z^\top = {\bm M}_z {\bm M}_x^\top$, which is referred to as the strong self-orthogonality (SSO) property. The matrix ${\bm M}$ is called the \emph{transfer} matrix. The following lemma summarizes the functionality of the $N$-sum box.

\begin{lemma}[$N$-sum box \cite{Allaix_N_sum_box}] \label{lem:box}
	There exists a set of orthogonal quantum states, denoted as $\{\ket{{\bm v}}_{\bm M}\}_{{\bm v}\in \mathbb{F}_q^{N\times 1}}$ defined on $\mathcal{H}_q^{\otimes N}$, the Hilbert space of $N$ $q$-dimensional quantum subsystems $Q_1,Q_2,\cdots, Q_N$, such that when ${\sf X}(x_i){\sf Z}(z_i)$ is applied to $Q_i$ for all $i \in [N]$, the state of the composite quantum system $Q$ changes from $\ket{\bm a}_{\bm M}$ to $\ket{{\bm a}+{\bm M}\bbsmatrix{{\bm x}\\{\bm z}}}_{\bm M}$ (with global phases omitted), i.e., $\otimes_{i\in [N]}{\sf X}(x_i){\sf Z}(z_i) \ket{{\bm a}}_{\bm M} \equiv \ket{{\bm a}+{\bm M}\bbsmatrix{{\bm x}\\{\bm z}}}_{\bm M}$, for all ${\bm x} \triangleq [x_1,\cdots, x_N]^\top \in \mathbb{F}_q^{N\times 1}$ and ${\bm z} \triangleq [z_1,\cdots, z_N]^\top \in \mathbb{F}_q^{N\times 1}$.
\end{lemma}
\noindent Note that each of these $q^N$ orthogonal quantum states is uniquely indexed by a vector in $\mathbb{F}_q^N$. According to the lemma, if the input state is chosen as $\ket{\bm 0}_{{\bm M}}$, then the output state is $\ket{{\bm M}\bbsmatrix{{\bm x}\\{\bm z}}}_{{\bm M}}$. Since the states are orthogonal, ${\bm y} = {\bm M}\bbsmatrix{{\bm x}\\{\bm z}}$ can be obtained with certainty by jointly measuring the quantum system $Q_1Q_2\cdots Q_N$ in the basis $\{\ket{{\bm v}}_{\bm M}\}_{{\bm v}\in \mathbb{F}_q^{N\times 1}}$.  

It is noteworthy that coding schemes based on the $N$-sum box have been shown to be capacity achieving for the $\Sigma$-QMAC (where the desired computation is simply a sum of the transmitters' inputs) with arbitrarily distributed entanglements in \cite{Yao_Jafar_Sum_MAC}, for the $\Sigma$-QEMAC, i.e., the $\Sigma$-QMAC where the channels are subject to erasures \cite{Yao_Jafar_SQEMAC}, and for several QPIR applications \cite{song_multiple_server_PIR,song_all_but_one_collusion}.

\subsection{Classical communication capacity of a noiseless quantum channel}\label{sec:classcom}
The classical \emph{communication} capacity of a point-to-point noisy quantum channel was studied in \cite{holevo1973bounds, holevo1998capacity, bennett_shor_capacity}, and the special case of a noiseless channel is particularly well understood (e.g., see \cite[Table I]{bennett_shor_capacity}). The noiseless channel capacity result is informally summarized as follows: 
\begin{enumerate}[align=left]
	\item[\textbf{Fact 1:}] Without receiver-side entanglement, a $\delta$-dimensional quantum system can carry at most $\log_d \delta$ dits of classical information; 
	\item[\textbf{Fact 2:}] With unlimited receiver-side entanglement, a $\delta$-dimensional quantum system can carry at most $2\log_d\delta $ dits of classical information. 
\end{enumerate}
For our \emph{computation} problem, the point to point \emph{communication} capacity results yield elementary \emph{converse bounds}  through cut-set arguments \cite{Appuswamy1}, i.e., by separating the parties into two groups and allowing full-cooperation within each group, collectively considering each group as the transmitter or the receiver, and bounding the communication costs in the resulting communication problem. Remarkably, while cut-set arguments were sufficient to obtain tight converse bounds in the $\Sigma$-QMAC \cite{Yao_Jafar_Sum_MAC}, these bounds will not suffice for the vector LC-QMAC  problem considered in this work.

\section{Results}
For $\mathcal{K} = \{k_1,k_2,\cdots, k_{|\mathcal{K}|}\} \subseteq [K]$, let us  define the following compact notations in Table \ref{tab:notations}, which will be useful in presenting our results.
\begin{table}[h]
\center
\begin{tabular}{rl}
\toprule
\textbf{Symbol} & \textbf{Description} 
\\ \midrule 
${\bm V}_{\mathcal{K}}:$ & $\begin{bmatrix}
	{\bm V}_{k_1} & {\bm V}_{k_2}& \cdots & {\bm V}_{k_{|\mathcal{K}|}} \end{bmatrix}  $ \\  \addlinespace[4pt]
$r_{\mathcal{K}}:$ & $\rk({\bm V}_{\mathcal{K}})$\\  \addlinespace[4pt]
$s_{\mathcal{K}}:$ & $\rk({\bm V}_{[K]}) - \rk({\bm V}_{[K] \setminus \mathcal{K}})$\\ \addlinespace[4pt]  
$\Delta_{\mathcal{K}}:$ & $\sum_{k\in \mathcal{K}} \Delta_k$ \\\bottomrule
\end{tabular}
\caption{Useful compact notations.}
\label{tab:notations}
\end{table}

\noindent For example, for $K=3$, ${\bm V}_{\{1,2\}} = \big[ {\bm V}_1, {\bm V}_2 \big]$, $r_{\{1,2\}} = \rk([{\bm V}_1, {\bm V}_2]), s_{\{1,2\}} = \rk([{\bm V}_1,{\bm V}_2,{\bm V}_3])-\rk({\bm V}_3)$, $\Delta_{[3]} = \Delta_1 + \Delta_2 + \Delta_3$.

\subsection{Converse bounds on $\mathfrak{D}^*$}
Let us first formalize for our LC-QMAC setting a baseline result that follows from existing work as mentioned  in Section \ref{sec:classcom}.
\begin{theorem}[Communication bounds] \label{thm:cut_set}
	The following bounds hold for the LC-QMAC$(\mathbb{F}_d, K,{\bm V}_1, \cdots, {\bm V}_K)$, 
	{\small
	\begin{align} \label{eq:cut_set_1}
		\Delta_{[K]} \geq r_{[K]},\\
		2\Delta_{\mathcal{K}} \geq r_{\mathcal{K}}, &&\forall \mathcal{K}\subseteq [K].\label{eq:cut_set_2}
	\end{align}
	}%
\end{theorem}
\noindent This theorem essentially follows from the known capacity results of quantum communication channels (e.g., \cite{bennett_shor_capacity}) together with a cut-set argument in network coding (e.g., \cite{Appuswamy1}). A formal proof is provided in Section \ref{proof:cut_set}. The following discussion elaborates upon the cut-set argument. 
\begin{enumerate}
	\item Consider Alice$_1$ -- Alice$_K$ together as one transmitter that has all the data and Bob as the receiver. The receiver must be able to recover ${\bm V}_1W_1^{[L]} + {\bm V}_2W_2^{[L]} +\cdots+ {\bm V}_KW_K^{[L]}$, which is  $L \times r_{[K]}$ dits of information. According to Fact 1 in Section \ref{sec:classcom}, $\log \delta_1+\log \delta_2+\cdots+\log \delta_K \geq L \times r_{[K]} \implies \Delta_{[K]}\geq r_{[K]}$. This gives us the bound  \eqref{eq:cut_set_1}.
	\item Let $\mathcal{K} \subseteq [K]$. Consider the Alices with indices in $\mathcal{K}$ collectively as the transmitter, and the rest of the Alices joining Bob together as the receiver (making their data and entangled quantum resource available to Bob for free). Then the receiver must be able to recover $\sum_{k\in \mathcal{K}}{\bm V}_kW_k^{[L]}$ from the merged transmitter. Note that what the receiver recovered constitutes $L \times r_{\mathcal{K}}$ dits of information. According to Fact 2 in Section \ref{sec:classcom}, we have  $2\sum_{k\in [K]}\log \delta_k \geq L \times r_{\mathcal{K}} \implies 2\Delta_{\mathcal{K}} \geq  r_{\mathcal{K}}$, which is the bound \eqref{eq:cut_set_2}.
\end{enumerate}

Next, as the first significant contribution of this work, we present the following stronger converse bounds.
\begin{theorem}[Multiparty computation bounds] \label{thm:new_bounds}
	Let $\{\mathcal{K}_1,\mathcal{K}_2,$ $ \cdots, $ $\mathcal{K}_T\}$ be a partition of $[K]$. Then the following bounds hold,
	{\small
	\begin{align}
		 \forall &T\geq 1,\notag \\
		&2\Delta_{[K]} \geq s_{\mathcal{K}_1} +  r_{\mathcal{K}_1}+ r_{\mathcal{K}_2}+ \cdots + r_{\mathcal{K}_T}, \label{eq:new_bound_1}\\
		 \forall & T\geq 2, \notag \\
		&2\big(\Delta_{\mathcal{K}_1} + \Delta_{\mathcal{K}_2} \big) 
		+ 4\big( \Delta_{\mathcal{K}_3}+\cdots \Delta_{\mathcal{K}_T}\big)		\notag \\
		& \geq   \big(s_{\mathcal{K}_1} + s_{\mathcal{K}_2}\big) +  \big( r_{\mathcal{K}_1}+ r_{\mathcal{K}_2}\big) +2 \big(r_{\mathcal{K}_3} + \cdots + r_{\mathcal{K}_T} \big). \label{eq:new_bound_2}
	\end{align} 
	}%
\end{theorem}
\noindent The proof of Theorem \ref{thm:new_bounds} is presented in Section \ref{proof:new_bounds}. Note that \eqref{eq:cut_set_1} is recovered as a special case of \eqref{eq:new_bound_1} by setting $T=1$, $\mathcal{K}_1=[K]$ which corresponds to $r_{[K]} = s_{[K]}$. Next let us illustrate the theorem with a couple of toy examples.\\

\subsubsection{Toy Example 4} 
To see how the converse bounds from Theorem \ref{thm:new_bounds} can be significantly stronger than those from  Theorem \ref{thm:cut_set}, consider the following example. Suppose $K\geq 2$ and ${\bm V}_1 = I_{K\times K}, ~ {\bm V}_k = {\bm D}_1, \forall k \in \{2,3,\cdots,K\}$, where $I_{K\times K}$ denotes the $K\times K$ identity matrix and ${\bm D}_1$ is the first column of $I_{K\times K}$.\footnote{For example, this setting includes the case of $K=3$ transmitters, namely  Alice$_1$, Alice$_2$, Alice$_3$, who have data $(A,B,C), (D),(E)$, respectively, and the receiver (Bob) desires the vector $(A+D+E, B,C)$.} It is not difficult to verify that the best bound implied by Theorem \ref{thm:cut_set} for the total download cost $\Delta_{[K]}$ is $\Delta_{[K]} \geq K$, whereas Theorem \ref{thm:new_bounds} implies $\Delta_{[K]} \geq 3K/2-1$. Thus, we note that the gap between the two bounds can be of the order of $K$. In other words, the additive gap between the baseline cut-set bounds of Theorem \ref{thm:cut_set} and the optimal value of the sum-download cost $\Delta_{[K]}$, is unbounded in general.

\subsubsection{Toy Example 5}
Consider an LC-QMAC with $K=4$ transmitters, namely Alice$_k$, $k\in[4]$.  Each Alice$_k$ has data $(x_k,z_k)$, say all symbols in $\mathbb{F}_3$, and sends one qudit $(d=3)$ to Bob. Then is  it possible for Bob to obtain 
{\small
\begin{align*}
{\bm y} =\bbsmatrix{x_1\\ x_2+x_3+ x_4\\ z_1\\ z_2+z_3+z_4}= {\bm M}\bbsmatrix{{\bm x}\\{\bm z}},
\end{align*} 
}%
where
{\small
\begin{align}
		&\hspace{1.5cm} {\bm M} = \bbsmatrix{
			1 & 0 & 0 & 0 & 0 & 0 & 0 & 0 \\
			0& 1 & 1 & 1 & 0 & 0 & 0 & 0 \\
			0 & 0 & 0 & 0 & 1 & 0 & 0 & 0\\
			0 & 0 & 0 & 0 & 0 & 1 & 1 & 1 \\
		},\notag \\
	&{\bm x} = [x_1,x_2,x_3,x_4]^\top, ~{\bm z} = [z_1,z_2,z_3,z_4]^\top,
\end{align}
}%
by measuring the four qudits? We cannot immediately construct such an $N$-sum box protocol because this ${\bm M}$ does not satisfy the SSO condition. But could  this be achieved through some other construction? Theorem \ref{thm:cut_set} does not preclude the existence of such a construction because the constraints \eqref{eq:cut_set_1} and \eqref{eq:cut_set_2} are not violated. However, Theorem \ref{thm:new_bounds} shows that such a computation is not possible by \emph{any} construction, i.e., it violates the laws of quantum physics. To see this, consider the $T=4$ way partition $(\mathcal{K}_1,\mathcal{K}_2,\mathcal{K}_3,\mathcal{K}_4)=(\{1\},\{2\},\{3\},\{4\})$. We have $s_{\{1\}}=2$ and $r_{\{t\}}=2$ for $t=1,2,3,4$, so  Condition \eqref{eq:new_bound_1} in Theorem \ref{thm:new_bounds} implies that $2\Delta_{[4]}\geq 10$, i.e., at least a total of $5$ qudits must be sent from the four Alices to Bob in order for Bob to recover such an output function.

\subsection{Capacity for $K=3$}
As the  main result of this work, we now characterize the capacity for LC-QMAC when $K=3$,  establishing in the process that the bounds from Theorem \ref{thm:cut_set} and Theorem \ref{thm:new_bounds} together provide a tight converse.
\begin{theorem} \label{thm:main_K3}
	For the LC-QMAC problem $(K=3, \mathbb{F}_d, {\bm V}_1,$ $ {\bm V}_2, {\bm V}_3)$, the optimal cost region $\mathcal{D}^*$ is the set of cost tuples $(\Delta_1,\Delta_2,\Delta_3) \in \mathbb{R}^3$ such that
	{\small
	\begin{align}  \label{eq:region}
	\left\{
	\begin{array}{rl}
		2\Delta_k \hspace{-0.2cm} & \geq r_k , ~\forall k\in [3]\\
		\Delta_{[3]} \hspace{-0.2cm} & \geq r_{\{1,2,3\}} \\
		2\Delta_{[3]} \hspace{-0.2cm} & \geq r_1+r_2+r_3 + s_k, ~\forall k\in [3]\\
		2\Delta_{[3]} + 2\Delta_k \hspace{-0.2cm} & \geq r_1+r_2+r_3+r_k \\ &~~~~+s_1+s_2 + s_3 - s_k, ~\forall k \in  [3]  
	\end{array}\right..
	\end{align}
	}%
\end{theorem}

\noindent Note that the first two bounds match the converse bounds in Theorem 1, and the last two bounds match the converse bounds in Theorem 2, when applied to $K=3$. The proof is divided into two parts. The direct part (achievability) is proved in Section \ref{proof:achi}. The converse is proved in Section \ref{proof:converse}. According to Theorem \ref{thm:main_K3}, $\mathfrak{D}^*$ is characterized by $10$ linear inequalities on $(\Delta_1, \Delta_2, \Delta_3)$ that appear in  Condition  \eqref{eq:region}, and thus the region $\mathfrak{D}^*$ is a polyhedron. For the achievability proof, it suffices to show that each of the corner points of the polyhedron is achievable, because the achievability of all other points then follows from a standard time-sharing argument. For the converse, we shall show that all 10 bounds hold for the cost tuple achieved by \emph{any} feasible coding scheme, based on Theorem \ref{thm:cut_set} and Theorem \ref{thm:new_bounds}.

\begin{remark} \label{rem:impossibility}
To see Toy Example $3$ in terms of the notation used for the problem formulation, note that we have $W_1 = A,W_2 = B,W_3 = [C,D]^\top$ and $f(A,B,C,D) = [A+B+C,D]^\top$. This corresponds to,
{\small
\begin{align}
	&{\bm V}_1 = \bbsmatrix{1\\0}, ~{\bm V}_2 = \bbsmatrix{1\\0}, ~{\bm V}_3 = \bbsmatrix{1&0\\0&1},  \notag \\
	&\mbox{and}~\begin{bmatrix}
			\rk({\bm V}_1)\\
			\rk({\bm V}_2)\\
			\rk({\bm V}_3)\\
			\rk([{\bm V}_1,{\bm V}_2])\\
			\rk([{\bm V}_1,{\bm V}_3])\\
			\rk([{\bm V}_2,{\bm V}_3])\\
			\rk([{\bm V}_1, {\bm V}_2, {\bm V}_3])
		\end{bmatrix}
		=
		\begin{bmatrix}
			1\\1\\2\\1\\2\\2\\2
		\end{bmatrix},
\end{align}
}%
which, by Theorem \ref{thm:main_K3}, produces the region $\mathfrak{D}^*$ specified in \eqref{eq:ex3}, and illustrated in Fig. \ref{fig:ex3}.
\end{remark}

\begin{remark} \label{rem:symmetric}
In the symmetric case where $\rk({\bm V}_1)=\rk({\bm V}_2)=\rk({\bm V}_3)\triangleq r_1$, $\rk([{\bm V}_1, {\bm V}_2])=\rk([{\bm V}_2, {\bm V}_3])$ $ =\rk([{\bm V}_3, {\bm V}_1]) \triangleq r_2$, and $\rk([{\bm V}_1, {\bm V}_2, {\bm V}_3])\triangleq r_3$, the optimal value of the total-download cost from Theorem \ref{thm:main_K3} is found to be $\max \{1.5r_1 + 0.75(r_3 - r_2), r_3\}$.
\end{remark}

\subsubsection{Toy Example 6}
As one more example, consider $K=3$ transmitters, and let $(A,B,C,D,E,F,G,H,I)$ be $9$ variables in a finite field $\mathbb{F}_d$. Let 
{\small
$$W_1 = [A,D,G]^\top, W_2 = [B,E,H]^\top, W_3 = [C,F,I]^\top$$} and {\small $$f(A,B,\cdots, I) = [A+B+C,~D+E+F,~ G,~H,~I]^\top.$$} From this, we obtain,
{\small
\begin{align}
	{\bm V}_1 = \bbsmatrix{1 & 0 & 0\\0 & 1 & 0\\0 & 0 & 1\\0 & 0 & 0\\ 0 & 0 & 0}, {\bm V}_2 = \bbsmatrix{1 & 0 & 0\\0 & 1 & 0\\0 & 0 & 0\\0 & 0 & 1\\0 & 0 & 0}, {\bm V}_3 = \bbsmatrix{1 & 0 & 0\\0 & 1 & 0\\0 & 0 & 0\\0 & 0 & 0\\0 & 0 & 1},
\end{align}
}%
{\small
\begin{align}
	\begin{bmatrix}
			\rk({\bm V}_1)\\
			\rk({\bm V}_2)\\
			\rk({\bm V}_3)\\
			\rk([{\bm V}_1,{\bm V}_2])\\
			\rk([{\bm V}_1,{\bm V}_3])\\
			\rk([{\bm V}_2,{\bm V}_3])\\
			\rk([{\bm V}_1, {\bm V}_2, {\bm V}_3])
		\end{bmatrix}
		=
		\begin{bmatrix}
			3\\3\\3\\4\\4\\4\\5
		\end{bmatrix},
\end{align}
}%
and
{\small
\begin{align}
	\mathfrak{D}^* = 
	\left\{
	\begin{bmatrix}
		\Delta_1\\\Delta_2\\\Delta_3
	\end{bmatrix}
	\in \mathbb{R}^3
	\left|
	\begin{array}{l}
		\Delta_1 \geq 3/2\\
		\Delta_2 \geq 3/2\\
		\Delta_3 \geq 3/2\\
		2\Delta_1+\Delta_2+\Delta_3 \geq 7\\
		\Delta_1+2\Delta_2+\Delta_3 \geq 7\\
		\Delta_1+\Delta_2+2\Delta_3 \geq 7
	\end{array}
	\right.
	\right\},
\end{align}
}%
which is illustrated in Fig. \ref{fig:ex6}.

\begin{figure}[htbp]
\center	
\includegraphics{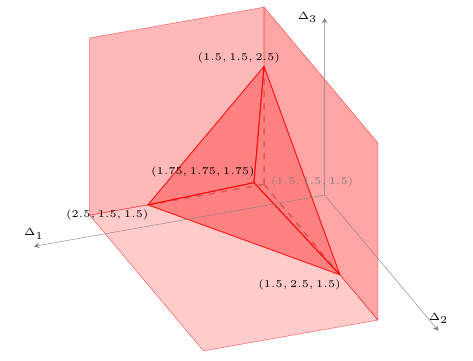}
\caption{$\mathfrak{D}^*$ for Toy Example 6.}\label{fig:ex6}
\end{figure}

\section{Proof of Converse Bounds}  \label{proof:converse}
In this section we present the proof for Theorem \ref{thm:cut_set} and Theorem \ref{thm:new_bounds}.
Consider any feasible LC-QMAC coding scheme with batch size $L$. Since the scheme must be correct for all realizations of 
{\small
\begin{align*}
	&(W_1^{[L]}, W_2^{[L]}, \cdots, W_K^{[L]})=(w_1,w_2,\cdots, w_K) \notag \\
	&\hspace{2cm} \in \mathbb{F}_d^{m_1\times L}\times\mathbb{F}_d^{m_2\times L}\times \cdots \times \mathbb{F}_d^{m_K\times L}, 
\end{align*}
}%
it must be correct even under the additional assumption that $(W_1^{[L]}, W_2^{[L]}, \cdots, W_K^{[L]})$ are generated uniformly  in $\mathbb{F}_d^{m_1\times L}\times\mathbb{F}_d^{m_2\times L}\times \cdots \times \mathbb{F}_d^{m_K\times L}$. Note that this assumption implies that $W_1^{[L]}, W_2^{[L]},  \cdots, W_K^{[L]}$ are independent. For compact notation, in the remainder of this section, we omit the superscript `$[L]$' over the data streams. Let $\rho$ denote the state of the joint classical-quantum system $W_1W_2\cdots W_K Q_1Q_2\cdots Q_K$ in the \textbf{encoding stage}.

\begin{lemma}[No-communication] \label{lem:no_communi}
	$I(W_{\mathcal{J}};W_{\mathcal{I}}Q_{\mathcal{I}})_\rho = 0$ for exclusive subsets $\mathcal{I}, \mathcal{J} \subseteq [K]$. Since conditional mutual information is non-negative, this directly implies that   $I(W_{\mathcal{J}};Q_{\mathcal{I}}\mid W_{\mathcal{I}})_{\rho} = 0$ and that $H(Q_{\mathcal{I}}\mid W_{\mathcal{I}}, W_{\mathcal{J}})_{\rho} = H(Q_{\mathcal{I}}\mid W_{\mathcal{I}})_{\rho}$.
\end{lemma}
\begin{proof}
	Since $W_1, W_2,\cdots, W_K$ are assumed independent, this implies $W_\mathcal{I}$ and $W_\mathcal{J}$ are independent. The lemma now follows from the no-communication theorem, e.g., see \cite{Peres_QIT}.
\end{proof}

\subsection{Proof of Theorem 1} \label{proof:cut_set}
We remind the readers of the definitions of ${\bm V}_{\mathcal{K}}, r_{\mathcal{K}}, s_{\mathcal{K}}$ and $\Delta_{\mathcal{K}}$ in Table \ref{tab:notations}.
Recall that in the decoding stage, Bob measures $Q_{[K]}$ to obtain $Y$, from which he gets $F^{[L]}$ (written as $F$ to simplify notation). 
Therefore, for any  $\mathcal{K}\subseteq [K]$,
{\small
\begin{align}
	&L\times r_{\mathcal{K}} \notag\\
	&=I(W_{\mathcal{K}};F\mid W_{[K] \setminus \mathcal{K}}) \\
	&= I(W_{\mathcal{K}};Y\mid W_{[K] \setminus \mathcal{K}}) \\
	&\leq I(W_{\mathcal{K}}; Q_{[K]} \mid W_{[K] \setminus \mathcal{K}})\label{eq:conv_1_1} \\
	&\leq H(Q_{[K]}) \label{eq:conv_1_2} \\
	&\leq \sum_{k\in [K]}\log \delta_k\\
	& \hspace{-0.3cm} \implies \Delta_{[K]} \geq r_{\mathcal{K}}
\end{align}
}%
Plugging in $\mathcal{K} = [K]$ proves \eqref{eq:cut_set_1}. Information measures on and after Step \eqref{eq:conv_1_1} are with respect to the state $\rho$.
Step \eqref{eq:conv_1_1} follows from Holevo bound, since Bob measures $Q_{[K]}$ to obtain $Y$. Step \eqref{eq:conv_1_2} is because $W_{\mathcal{K}}$ and $W_{[K]\setminus \mathcal{K}}$ are classical, and thus conditioning on any realization of $W_{[K]\setminus \mathcal{K}}$, the mutual information between $W_{\mathcal{K}}$ and $Q_{[K]}$ is not greater than $H(Q_{[K]})$.

Continuing from \eqref{eq:conv_1_1},
{\small
\begin{align}
	& L \times r_{\mathcal{K}} \notag \\
	& \leq I(W_{\mathcal{K}}; Q_{[K]} \mid W_{[K]\setminus \mathcal{K}}) \notag \\
	& = I(W_{\mathcal{K}}; Q_{\mathcal{K}}\mid Q_{[K]\setminus \mathcal{K}}, W_{[K]\setminus \mathcal{K}}) \label{eq:conv_1_3} \\
	& = H(Q_{\mathcal{K}}\mid Q_{[K]\setminus \mathcal{K}},W_{[K]\setminus \mathcal{K}}) - H(Q_{\mathcal{K}}\mid Q_{[K]\setminus \mathcal{K}},W_{[K]}) \\
	& \leq H(Q_{\mathcal{K}}\mid Q_{[K]\setminus \mathcal{K}},W_{[K]\setminus \mathcal{K}}) + H(Q_{\mathcal{K}}\mid W_{\mathcal{K}}) \label{eq:conv_1_4} \\
	& \leq 2H(Q_{\mathcal{K}}) \label{eq:conv_1_5} \\
	& \leq 2\sum_{k\in \mathcal{K}} \log \delta_k \\
	& \hspace{-0.3cm} \implies  2\Delta_{\mathcal{K}} \geq r_{\mathcal{K}}
\end{align}
}%
This proves \eqref{eq:cut_set_2}. Step \eqref{eq:conv_1_3} follows from Lemma \ref{lem:no_communi}, which implies $I(W_{\mathcal{K}};Q_{[K] \setminus \mathcal{K}} \mid W_{[K]\setminus \mathcal{K}}) = 0$. Step \eqref{eq:conv_1_4} follows from the Araki-Lieb triangle inequality, by conditioning on $W_{[K]}$, and  noting that $H(Q_{\mathcal{K}} \mid W_{[K]}) = H(Q_{\mathcal{K}} \mid W_{\mathcal{K}})$, as implied by Lemma \ref{lem:no_communi}. Step \eqref{eq:conv_1_5} holds because conditioning does not increase entropy.

\subsection{Proof of Theorem 2} \label{proof:new_bounds}
We need the following lemmas.
\begin{lemma} \label{lem:comp1}
For $\mathcal{K} \subseteq [K]$,
{\small
\begin{align}
	L \times s_{\mathcal{K}} \leq H(Q_{[K]}) - H(Q_{[K]}\mid W_{\mathcal{K}}).
\end{align}
}%
\end{lemma}
\begin{proof}
{\small
\begin{align}
	& L \times s_{\mathcal{K}} \notag \\
	& = L \times (r_{[K]} - r_{[K] \setminus \mathcal{K}}) \notag\\
	& = H(F) - H\left(\sum_{k\in [K] \setminus \mathcal{K}}{\bm V}_kW_k \right) \\
	& = H(F) - H(F\mid W_{\mathcal{K}}) \\
	& = I(F;W_{\mathcal{K}})\\
	& = I(Y;W_{\mathcal{K}})\\
	& \leq  I(Q_{[K]}; W_{\mathcal{K}}) \label{eq:comp1_1}  \\
	& = H(Q_{[K]}) - H(Q_{[K]}\mid W_{\mathcal{K}}) 
\end{align}
}%
Step \eqref{eq:comp1_1} follows from Holevo's bound.
\end{proof}

\begin{lemma} \label{lem:comp2}
	For $\{\mathcal{K}_1,\mathcal{K}_2, \cdots, \mathcal{K}_T\}$ a partition of $[K]$,
	{\small
	\begin{align}
		&L \times (s_{\mathcal{K}_1} + r_{\mathcal{K}_2} + \cdots + r_{\mathcal{K}_T})\notag \\
		& \leq H(Q_{[K]}) - H(Q_{\mathcal{K}_1}\mid W_{\mathcal{K}_1}) + \sum_{i=2}^T H(Q_{\mathcal{K}_i}\mid W_{\mathcal{K}_i}). \label{eq:comp2}
	\end{align}
	}%
\end{lemma}
\begin{proof}
According to Lemma \ref{lem:comp1} and \eqref{eq:conv_1_4}, we have
{\small
	\begin{align}
		& L \times (s_{\mathcal{K}_1} + r_{\mathcal{K}_2} + \cdots + r_{\mathcal{K}_T}) \notag \\
		& \leq H(Q_{[K]}) - H(Q_{[K]}\mid W_{\mathcal{K}_1}) \notag \\
		& ~~~~ + \sum_{i=2}^T \Big( H(Q_{\mathcal{K}_i} \mid Q_{[K]\setminus \mathcal{K}_i}, W_{[K]\setminus\mathcal{K}_i}) + H(Q_{\mathcal{K}_i}\mid W_{\mathcal{K}_i}) \Big) \label{eq:comp2_0}\\
		& \leq H(Q_{[K]}) - H(Q_{[K]}\mid W_{\mathcal{K}_1}) +  \sum_{i=2}^T H(Q_{\mathcal{K}_i}\mid W_{\mathcal{K}_i})  \notag \\
		& ~~~~ + \sum_{i=2}^T \Big( H(Q_{\mathcal{K}_i} \mid Q_{\mathcal{K}_1 \cup \cdots \cup \mathcal{K}_{i-1} \cup \mathcal{K}_{i+1} \cup \cdots \cup \mathcal{K}_T}, W_{\mathcal{K}_1}) \Big)   \label{eq:comp2_1}\\
		& \leq H(Q_{[K]}) - H(Q_{\mathcal{K}_1}\mid W_{\mathcal{K}_1}) +  \sum_{i=2}^T H(Q_{\mathcal{K}_i}\mid W_{\mathcal{K}_i}) \notag \\&~~~~- H(Q_{\mathcal{K}_2\cup \mathcal{K}_3 \cup \cdots \cup \mathcal{K}_T}\mid Q_{\mathcal{K}_1}, W_{\mathcal{K}_1}) \notag \\
		& ~~~~ + \sum_{i=2}^T \Big( H(Q_{\mathcal{K}_i} \mid Q_{\mathcal{K}_1 \cup \cdots \cup \mathcal{K}_{i-1}}, W_{\mathcal{K}_1}) \Big) \label{eq:comp2_2} \\
		& = H(Q_{[K]}) - H(Q_{\mathcal{K}_1}\mid W_{\mathcal{K}_1}) + \sum_{i=2}^T H(Q_{\mathcal{K}_i}\mid W_{\mathcal{K}_i}) \label{eq:comp2_3}
	\end{align}
	}%
where in Steps \eqref{eq:comp2_1} and \eqref{eq:comp2_2} we use the fact that conditioning does not increase entropy. Step \eqref{eq:comp2_3} follows from the chain rule for entropy.
\end{proof}
The series of inequalities that appear in the proofs of these two lemmas are not expected to be tight in general, but they suffice to derive the desired converse bounds, namely \eqref{eq:new_bound_1} and \eqref{eq:new_bound_2} in Theorem \ref{thm:new_bounds}, as shown next.

We proceed as follows. First, by Lemma \ref{lem:comp2} and \eqref{eq:conv_1_4}, we have
{\small
\begin{align}
	& L \times (s_{\mathcal{K}_1}+r_{\mathcal{K}_1}+r_{\mathcal{K}_2}+\cdots + r_{\mathcal{K}_T}) \notag \\
	& \leq H(Q_{[K]}) + H(Q_{\mathcal{K}_1} \mid Q_{\mathcal{K}_2 \cup \cdots \cup \mathcal{K}_T}, W_{\mathcal{K}_2 \cup \cdots \cup \mathcal{K}_T})\notag \\
	&~~~~+ H(Q_{\mathcal{K}_2}\mid W_{\mathcal{K}_2}) + \cdots + H(Q_{\mathcal{K}_T}\mid W_{\mathcal{K}_T}) \\
	& \leq 2 \sum_{k\in [K]} \log  \delta_k\\
	& \hspace{-0.3cm} \implies \eqref{eq:new_bound_1}
\end{align}
}%
Next, noting the symmetry in Lemma \ref{lem:comp2}, we have
{\small
\begin{align}
	&L \times (s_{\mathcal{K}_2} + r_{\mathcal{K}_1} +  r_{\mathcal{K}_3} + \cdots + r_{\mathcal{K}_T})\notag \\
	& \leq H(Q_{[K]}) - H(Q_{\mathcal{K}_2}\mid W_{\mathcal{K}_2}) +  H(Q_{\mathcal{K}_1} \mid W_{\mathcal{K}_1}) \notag \\
	&~~~~ +  H(Q_{\mathcal{K}_3} \mid W_{\mathcal{K}_3}) + \cdots + H(Q_{\mathcal{K}_T} \mid W_{\mathcal{K}_T}). \label{eq:comp2_sym}
\end{align}
}%
Adding \eqref{eq:comp2} and \eqref{eq:comp2_sym}, we obtain
{\small
\begin{align}
	& L \times \big(s_{\mathcal{K}_1} + s_{\mathcal{K}_2}\big) +  \big( r_{\mathcal{K}_1}+ r_{\mathcal{K}_2}\big) + 2 \big(r_{\mathcal{K}_3} + \cdots + r_{\mathcal{K}_T} \big) \notag \\
	& \leq 2H(Q_{[K]}) + 2\big( H(Q_{\mathcal{K}_3}\mid W_{\mathcal{K}_3}) + \cdots + H(Q_{\mathcal{K}_T}\mid W_{\mathcal{K}_T}) \big) \\
	& \leq 2\sum_{k\in \mathcal{K}_1\cup \mathcal{K}_2} \log \delta_{k} + 4\sum_{k\in \mathcal{K}_3\cup\cdots \cup \mathcal{K}_T} \log \delta_{k} \\
	& \hspace{-0.3cm} \implies \eqref{eq:new_bound_2}
\end{align}
}%
 
\section{Proof of Theorem 3: Achievability} \label{proof:achi}
\subsection{Standard form of the linear function}
Given the LC-QMAC problem specified by $(\mathbb{F}_d, K=3, {\bm V}_1, {\bm V}_2, {\bm V}_3)$, the function computed at Bob is by definition,
{\small
\begin{align} \label{eq:function_original}
	F = {\bm V}_1 W_1 + {\bm V}_2 W_2 + {\bm V}_3 W_3.
\end{align}
}%
According to \cite[Lemma 2]{Yao_Jafar_3LCBC}, there exist $\mathbb{F}_d$ matrices (with full column ranks and $m$ rows each) $$\{U_{123}, U_{12}, U_{13}, U_{23}, U_{1(2,3)}, U_{2(1,3)}, U_{3(1,2)}, U_{1}, U_{2}, U_{3}\}$$ such that,
{\small
\begin{enumerate}
	\item $\begin{bmatrix} U_{123} & U_{12} & U_{13} & U_{1(2,3)} & U_{1} \end{bmatrix}$ form a basis for the column span of ${\bm V}_1$;
	\item $\begin{bmatrix} U_{123} & U_{12} & U_{23} & U_{2(1,3)} & U_{2} \end{bmatrix}$ form a basis for the column span of ${\bm V}_2$;
	\item $\begin{bmatrix} U_{123} & U_{13} & U_{23} & U_{3(1,2)} & U_{3} \end{bmatrix}$ form a basis for the column span of ${\bm V}_3$;
	\item $\begin{bmatrix} U_{123} & U_{12} & U_{13} & U_{23} & U_{1(2,3)} & U_{2(1,3)} & U_{1} & U_{2} \end{bmatrix}$ form a basis for the column span of $[{\bm V}_1, {\bm V}_2]$;
	\item $\begin{bmatrix} U_{123} & U_{12} & U_{13} & U_{23} & U_{1(2,3)} & U_{3(1,2)} & U_{1} & U_{3} \end{bmatrix}$ form a basis for the column span of $[{\bm V}_1, {\bm V}_3]$;
	\item $\begin{bmatrix} U_{123} & U_{12} & U_{13} & U_{23} & U_{2(1,3)} & U_{3(1,2)} & U_{2} & U_{3} \end{bmatrix}$ form a basis for the column span of $[{\bm V}_2, {\bm V}_3]$;
	\item $\begin{bmatrix} U_{123} & U_{12} & U_{13} & U_{23} & U_{2(1,3)} & U_{3(1,2)} & U_{1} & U_{2} & U_{3} \end{bmatrix}$ form a basis for the column span of $[{\bm V}_1, {\bm V}_2, {\bm V}_3]$.
	\item $U_{1(2,3)}, U_{2(1,3)}$ and $U_{3(1,2)}$ have the same size and $U_{1(2,3)} = U_{2(1,3)} + U_{3(1,2)}$.
\end{enumerate}
}%
Let $n_*$ denote the number of columns of $U_*$, for $*\in \{1,2,3,12,13,23,123\}$. The number of columns for $U_{1(2,3)}$ (the same for $U_{2(1,3)}$ and $U_{3(1,2)}$) is denoted as $n_o$.

Recall that each ${\bm V}_k$ is an $m\times m_k$ matrix. Since we assume without loss of generality that $m_k\leq m$ and each ${\bm V}_k$ has full column rank, it follows that there exist invertible matrices $R_1, R_2, R_3$ such that
{\small
\begin{align}
	{\bm V}_1 &= \begin{bmatrix}
		U_{123} & U_{12} & U_{13} & U_{1(2,3)} & U_{1}
	\end{bmatrix}
	R_1,\\
	{\bm V}_2 &= \begin{bmatrix}
		U_{123} & U_{12} & U_{23} & U_{2(1,3)} & U_{2}
	\end{bmatrix}
	R_2,\\
	{\bm V}_3 &= \begin{bmatrix}
		U_{123} & U_{13} & U_{23} & U_{3(1,2)} & U_{3}
	\end{bmatrix}
	R_3.
\end{align}
}%
Thus, \eqref{eq:function_original} becomes
{\small
\begin{align} \label{eq:function_intermediate}
	&F = \begin{bmatrix}
		U_{123} & U_{12} & U_{13} & U_{1(2,3)} & U_{1}
	\end{bmatrix}
	(R_1W_1) \notag\\
	&~~~~ + \begin{bmatrix}
		U_{123} & U_{12} & U_{23} & U_{2(1,3)} & U_{2}
	\end{bmatrix}
	(R_2W_2) \notag\\
	&~~~~+ \begin{bmatrix}
		U_{123} & U_{13} & U_{23} & U_{3(1,2)} & U_{3}
	\end{bmatrix}
	(R_3W_3).
\end{align}
}%
$R_kW_k$ can be considered as the ($m_k$-dimensional) data available to Alice$_k$ for $k\in [3]$. It will be convenient to write $R_kW_k$ as,
{\small
\begin{align}
	R_1W_1 = \begin{bmatrix}
		A_{123} \\ A_{12} \\ A_{13} \\ A_{o} \\ A_{1}
	\end{bmatrix},
	R_2W_2 = \begin{bmatrix}
		B_{123} \\ B_{12} \\ B_{23} \\ B_{o} \\ B_{2}
	\end{bmatrix},
	R_3W_3 = \begin{bmatrix}
		C_{123} \\ C_{13} \\ C_{23} \\ C_{o} \\ C_{3}
	\end{bmatrix},
\end{align}
}%
where $A_{123}, A_{12}, A_{13}, \cdots, C_o, C_3$ are vectors with elements drawn in $\mathbb{F}_d$, with $X_{*}$ being an $n_*$-length vector for $X \in \{A,B,C\}$ and $*\in \{o,1,2,3,12,13,23,123\}$. Then, \eqref{eq:function_intermediate} becomes,
{\footnotesize
\begin{align}
	&F = \underbrace{\begin{bmatrix}
		U_{123} & U_{12} & U_{13} & U_{23} & U_{2(1,3)} & U_{3(1,2)} & U_{1} & U_{2} & U_{3}
	\end{bmatrix}}_{\bm U}\notag \\
	&\hspace{2.75cm} \times \begin{bmatrix}
	A_{123}+B_{123}+C_{123} \\
	A_{12}+B_{12} \\
	A_{13}+C_{13} \\
	B_{23}+C_{23} \\
	A_o+B_o \\
	A_o+C_o \\
	A_1\\
	B_2\\
	C_3
	\end{bmatrix}
\end{align}
}%
by noting that $U_{1(2,3)} = U_{2(1,3)} + U_{3(1,2)}$. Since ${\bm U}$ is a basis (and thus has full column rank), computing $F$ is equivalent to computing $\tilde{F}$, where,
{\small
\begin{align} \label{eq:function_standard}
	\tilde{F} = \begin{bmatrix}
	A_{123}+B_{123}+C_{123} \\
	A_{12}+B_{12} \\
	A_{13}+C_{13} \\
	B_{23}+C_{23} \\
	A_o+B_o \\
	A_o+C_o \\
	A_1\\
	B_2\\
	C_3
	\end{bmatrix},
\end{align}
}%
such that all $A_{*}$ symbols come from Alice$_1$, $B_{*}$ come from Alice$_2$, and $C_{*}$ come from Alice$_3$. 
Let us refer to the form in \eqref{eq:function_standard} as the \emph{standard} form of the linear computation for $K=3$. We will refer to the elements of $\tilde{F}$ as demands, that the achievable scheme will need to satisfy. For example, the achievable scheme should satisfy $n_{123}$ dimensions of Bob's demands along $A_{123}+B_{123}+C_{123}$.

Specifically, the standard form is composed of,
\begin{enumerate}
	\item $A_{123}+B_{123}+C_{123}$, which is an $n_{123}$-dimensional $3$-way sum and each term comes from a different Alice;
	\item $A_{12}+B_{12}$, which is an $n_{12}$-dimensional $2$-way sum of the inputs from Alice$_1$ and Alice$_2$; $A_{13}+C_{13}$, which is an $n_{13}$-dimensional $2$-way sum of the inputs from Alice$_1$ and Alice$_3$; $B_{23}+C_{23}$, which is an $n_{23}$-dimensional $2$-way sum of the inputs from Alice$_2$ and Alice$_3$, such that $A_{12}, A_{13}, B_{12}, B_{23}, C_{13}, C_{23}$ are different terms.
	\item $A_o+B_o$, which is an $n_{o}$-dimensional $2$-way sum of the inputs from Alice$_1$ and Alice$_2$; $A_o+C_o$, which is another $n_{o}$-dimensional $2$-way sum of the inputs from Alice$_1$ and Alice$_3$, such that the same $A_o$ appears in both $A_o+B_o$ and $A_o+C_o$. Since $A_o+B_o$ and $A_o+C_o$ always have the same dimension $n_o$, in the following they shall always be considered together.
	\item $A_1$, an $n_1$-dimensional vector from Alice$_1$; $B_2$, an $n_2$-dimensional vector from Alice$_2$, and $C_3$, an $n_3$-dimensional vector from Alice$_3$.
\end{enumerate}
With this form, we can  evaluate Theorem \ref{thm:main_K3} in terms of $\{n_*\mid * \in \{1,2,3,12,13,23,123,o\}\}$ as
{\small
\begin{align} \label{eq:region_standard}
		\mathfrak{D}^* = \left\{
		\begin{bmatrix}
			\Delta_1\\\Delta_2\\\Delta_3	
		\end{bmatrix}
 		\in \mathbb{R}^3 \left|
		{\bm A} 
		\begin{bmatrix}
			\Delta_1\\\Delta_2\\\Delta_3
		\end{bmatrix}
		\geq
		{\bm B} 
		\begin{bmatrix}
			n_{123}\\
			n_{12}\\
			n_{13}\\
			n_{23}\\
			n_o\\
			n_1\\
			n_2\\
			n_3
		\end{bmatrix}
		\right.
		\right\},
\end{align}
}%
where 
{\small
\begin{align} \label{eq:def_AB}
	{\bm A} = \begin{bmatrix}
			2 & 0 & 0 \\
			0 & 2 & 0 \\
			0 & 0 & 2 \\
			1 & 1 & 1 \\
			2 & 2 & 2 \\
			2 & 2 & 2 \\
			2 & 2 & 2 \\
			4 & 2 & 2 \\
			2 & 4 & 2\\
			2 & 2 & 4
		\end{bmatrix},
	{\bm B} = \begin{bmatrix}
		1 & 1 & 1 & 0 & 1 & 1 & 0 & 0\\
		1 & 1 & 0 & 1 & 1 & 0 & 1 & 0\\
		1 & 0 & 1 & 1 & 1 & 0 & 0 & 1\\
		1 & 1 & 1 & 1 & 2 & 1 & 1 & 1\\
		3 & 2 & 2 & 2 & 3 & 2 & 1 & 1\\
		3 & 2 & 2 & 2 & 3 & 1 & 2 & 1\\
		3 & 2 & 2 & 2 & 3 & 1 & 1 & 2\\
		4 & 3 & 3 & 2 & 4 & 2 & 2 & 2\\
		4 & 3 & 2 & 3 & 4 & 2 & 2 & 2\\
		4 & 2 & 3 & 3 & 4 & 2 & 2 & 2
	\end{bmatrix}.
\end{align}
}%
In the remainder of this section, the goal is to prove that  $(\mbox{RHS of } \eqref{eq:region_standard}) \subseteq \mathfrak{D}^*$.

\subsection{Building block protocols}
Let us list the building block protocols that will be used to establish the achievable region.
The first building block protocol is based on trivially treating qudits as classical dits (TQC). It is summarized as follows.

\begin{enumerate}[align=left]
	\item[\textbf{[TQC]:}] For any transmitter with $\mathbb{F}_d$ input $x$, by receiving one (encoded) qudit from that transmitter, the receiver can measure $x$ with certainty. This protocol is suitable for satisfying Bob's demands along certain dimensions of $A_1, B_2, C_3$ in \eqref{eq:function_standard}. Specifically, when applying TQC to satisfy certain dimensions of $A_1$, the protocol is referred to as \textbf{P1}. An amortized cost tuple $(1,0,0)$ is used for this protocol as with (Alice$_1$, Alice$_2$, Alice$_3$) sending $(1,0,0)$ qudit, one dimension of $A_1$ demand is satisfied for Bob.  Similarly, \textbf{P2} with amortized cost tuple $(0,1,0)$ refers to TQC for satisfying $B_2$, and \textbf{P3} with amortized cost tuple $(0,0,1)$ refers to TQC for satisfying a $C_3$ demand.
\end{enumerate}
The rest of the building block protocols are based on the $N$-sum box (Lemma \ref{lem:box}). Note that when a protocol utilizes multiple $N$-sum boxes, the total cost is simply the sum of the costs of the boxes that are used.

To apply Lemma \ref{lem:box} one must first specify the transfer matrix ${\bm M} = [{\bm M}_x, {\bm M}_z]$ with full row rank $N$ and ${\bm M}_x {\bm M}_z^\top = {\bm M}_x {\bm M}_z^\top$ (SSO property). For our purpose, we need the following two $N$-sum boxes.
\begin{enumerate}[align=left]
	\item[\textbf{Box 1}:]  A $2$-sum box with transfer matrix {\small ${\bm M}_1 = \begin{bmatrix} 1 & 1 & 0 & 0 \\ 0 & 0 & 1 & -1 \end{bmatrix}$}.
	\item[\textbf{Box 2}:] A $3$-sum box with transfer matrix {\small ${\bm M}_2 = \begin{bmatrix}
				1 & 1 & 1 & 0 & 0 & 0 \\
				0 & 0 & 0 & 1 & -1 & 0\\
				0 & 0 & 0 & 1 & 0 & -1
			\end{bmatrix}$}.
\end{enumerate}
 It is readily verified that ${\bm M}_1, {\bm M}_2$ satisfy the SSO property. Using Box 1 with transfer matrix ${\bm M}_1$, we develop the following protocols.
\begin{enumerate}[align=left]
	\item[\textbf{[2-way-sums]:}]  For $2$ transmitters with $\mathbb{F}_d$ inputs $(x_1,z_1)$, and $(x_2,z_2)$, respectively, by receiving one qudit from each transmitter, the receiver can measure two sums $(x_1+x_2, z_1+z_2)$ with certainty. The negative sign can be handled by transmitter-side local operations. This protocol is suitable for satisfying demands along certain dimensions of $A_{12}+B_{12}$, referred to as \textbf{P4} with amortized cost tuple $(0.5,0.5,0)$, or $A_{13}+C_{13}$ (\textbf{P5} with amortized cost tuple $(0.5,0,0.5)$), or $B_{23}+C_{23}$ (\textbf{P6} with amortized cost tuple $(0,0.5,0.5)$). In addition, it is used to satisfy certain demands along the dimensions of $A_o+B_o$ or $A_o+C_o$. It will be sufficient to use 2-way-sums to satisfy demands for the \emph{same} number of dimensions in $A_o+B_o$ as in $A_o+C_o$. Specifically, by letting (Alice$_1$, Alice$_2$, Alice$_3$) send $(1,1,0)+(1,0,1)=(2,1,1)$ qudits, we satisfy $2$ demand dimensions in each of $A_o+B_o$ and $A_o+C_o$. The amortized cost tuple is $(1,0.5,0.5)$ per dimension of $A_o+B_o$ and $A_o+C_o$. Denote this protocol as \textbf{P7}. Then note that $(A_o+B_o, A_o+C_o)$, $(A_o+B_o, -B_o+C_o)$ and $(A_o+C_o, B_o-C_o)$ are computationally equivalent (invertible) expressions, i.e., any one of them suffices to compute all three of them. Therefore, alternatively, by sending $(1,2,1)$, or $(1,1,2)$ qudits, (Alice$_1$, Alice$_2$, Alice$_3$)  can also satisfy $2$ dimensions in both $A_o+B_o$ and $A_o+C_o$. This gives us another two protocols \textbf{P8} and \textbf{P9}, with respective amortized cost tuples $(0.5,1,0.5)$ and $(0.5,0.5,1)$.
	\item[\textbf{[Superdense coding]:}] Setting $x_2=z_2=0$ in 2-way-sums, by receiving one qudit from each of the two transmitters, the receiver can measure $(x_1, z_1)$ with certainty. Note that this is exactly the superdense coding protocol, and the second transmitter only provides entangled qudits. This protocol is suitable for satisfying demands along certain dimensions of  $A_1$ (referred to as \textbf{P10} if Alice$_2$ provides the entanglement, or \textbf{P11} if Alice$_3$ provides the entanglement). The amortized cost tuple for \textbf{P10} is $(0.5,0.5,0)$ per dimension of $A_1$, and for \textbf{P11} is $(0.5,0,0.5)$. Similarly we define \textbf{P12} with amortized cost tuple $(0.5,0.5,0)$, \textbf{P13} with amortized cost tuple $(0,0.5,0.5)$ as the protocols that use superdense coding to satisfy each dimension of $B_2$, and define \textbf{P14} with amortized cost tuple $(0.5,0,0.5)$, \textbf{P15} with amortized cost tuple $(0,0.5,0.5)$ as the protocols that use superdense coding to satisfy each dimension of $C_3$.
	\item[\textbf{[3-way-sums]:}] For $3$ transmitters with $\mathbb{F}_d$ inputs $(u_1, v_1, w_1, x_1)$, $(u_2,v_2,w_2,x_2)$ and $(u_3,v_3,w_3,$ $x_3)$, respectively, by applying 2-way-sums once to each pair of the three transmitters, with appropriate precoding at the transmitters,  the receiver obtains $[(u_1-v_1)+u_2, (w_1-x_1)+w_2]$, $[v_2+(v_3-u_3), x_2+(x_3-w_3)]$ and $[v_1+u_3, x_1+w_3]$. From these, it is easy to verify that the receiver is able to obtain $[u_1+u_2+u_3, v_1+v_2+v_3, w_1+w_2+w_3, x_1+x_2+x_3]$. In the process, the receiver receives $6$ qudits, $2$ from each transmitter, and the output is $4$ dimensions of $3$-way sums. This protocol is suitable for satisfying demands along certain dimensions of $A_{123}+B_{123}+C_{123}$. Note that the amortized cost tuple is $(0.5,0.5,0.5)$ per dimension of 3-way-sum.  Denote this protocol as \textbf{P16}.
\end{enumerate}
Using Box 2 with transfer matrix $\textbf{M}_2$, we further develop the following protocols.
\begin{enumerate}[align=left]
	\item[\textbf{[3+2+2]:}] For $3$ transmitters with  with inputs $(x_1,z_1)$, $(x_2,z_2)$ and $(x_3,z_3)$ respectively, by receiving one qudit from each transmitter, the receiver can measure three sums $(x_1+x_2+x_3, z_1+z_2, z_1+z_3)$ with certainty. Note that the same $z_1$ appears in both 2-way sums. This protocol is suitable for satisfying demands along certain dimensions of $(A_{123}+B_{123}+C_{123}, A_o+B_o, A_o+C_o)$. The amortized cost tuple is $(1,1,1)$ per dimension in each of $A_{123}+B_{123}+C_{123}, A_o+B_o$ and $A_o+C_o$. Denote this protocol as \textbf{P17}.
	\item[\textbf{[3+1+1]}:] Setting $z_1=0$ in \textbf{P5} allows the receiver to measure $(x_1+x_2+x_3, z_2, z_3)$ with certainty. This protocol is useful for satisfying demands along certain dimensions of $(A_{123}+B_{123}+C_{123}, A_1, B_2)$ (referred to as \textbf{P18}), or $(A_{123}+B_{123}+C_{123}, A_1, C_3)$ (referred to as \textbf{P19}), or $(A_{123}+B_{123}+C_{123}, B_2, C_3)$ (referred to as \textbf{P20}). The amortized cost tuple is $(1,1,1)$ for \textbf{P18}--\textbf{P20}.
\end{enumerate}

\subsection{Achievable region with auxiliary variables}
Define $\mathfrak{D}_{\achi}$ (on the top of the next page) in \eqref{eq:region_achievable} .
\begin{figure*}[t]
{\footnotesize
\begin{align} \label{eq:region_achievable}
	\mathfrak{D}_{\achi} = \left\{
		\begin{bmatrix}
			\Delta_1\\\Delta_2\\\Delta_3	
		\end{bmatrix}
 		\in \mathbb{R}^3 \left|
		\underbrace{\begin{bmatrix}
			\Delta_1\\\Delta_2\\\Delta_3
		\end{bmatrix}
		\geq
		{\bm C} 
		\begin{bmatrix}
			\lambda_1\\ \lambda_2 \\ \lambda_3 \\ \vdots \\ \lambda_{20}
		\end{bmatrix},
		{\bm D}
		\begin{bmatrix}
			\lambda_1\\ \lambda_2 \\ \lambda_3 \\ \vdots \\ \lambda_{20}
		\end{bmatrix}
		\geq
		\begin{bmatrix}
			n_{123}\\ n_{12} \\ n_{13} \\ n_{23} \\ n_o \\ n_1 \\ n_2 \\ n_3
		\end{bmatrix}}_{\mbox{\fontsize{6}{0} \textbf{Cond1}}},
		\underbrace{
		\begin{bmatrix}
			0 \\ 0 \\ 0 \\ \vdots \\ 0
		\end{bmatrix}
		\leq 
		\begin{bmatrix}
			\lambda_1\\ \lambda_2 \\ \lambda_3 \\ \vdots \\ \lambda_{20}
		\end{bmatrix} \in \mathbb{Q}^{20}}_{\mbox{\fontsize{6}{0} \textbf{Cond2}}(\mathbb{Q})}
		\right.
		\right\},
\end{align}
}%
{\footnotesize
\begin{align*}
	{\bm C} =  
	\left[
	\begin{array}{@{}*{20}{c}@{}}
		1 & 0 & 0 & 0.5 & 0.5 & 0 & 1 & 0.5 & 0.5 & 0.5 & 0.5 & 0.5 & 0 & 0.5 & 0 & 0.5 & 1 & 1 & 1 & 1 \\
		0 & 1 & 0 & 0.5 & 0 & 0.5 & 0.5 & 1 & 0.5 & 0.5 & 0 & 0.5 & 0.5 & 0 & 0.5 & 0.5 & 1 & 1 & 1 & 1 \\
		0 & 0 & 1 & 0 & 0.5 & 0.5 & 0.5 & 0.5 & 1 & 0 & 0.5 & 0 & 0.5 & 0.5 & 0.5 & 0.5 & 1 & 1 & 1 & 1
	\end{array}
	\right],
\end{align*}
\begin{align}
	{\bm D} =
	\left[
	\begin{array}{@{}*{20}{c}@{}}
		0 & 0 & 0 & 0 & 0 & 0 & 0 & 0 & 0 & 0 & 0 & 0 & 0 & 0 & 0 & 1 & 1 & 1 & 1 & 1 \\
		0 & 0 & 0 & 1 & 0 & 0 & 0 & 0 & 0 & 0 & 0 & 0 & 0 & 0 & 0 & 0 & 0 & 0 & 0 & 0 \\
		0 & 0 & 0 & 0 & 1 & 0 & 0 & 0 & 0 & 0 & 0 & 0 & 0 & 0 & 0 & 0 & 0 & 0 & 0 & 0 \\
		0 & 0 & 0 & 0 & 0 & 1 & 0 & 0 & 0 & 0 & 0 & 0 & 0 & 0 & 0 & 0 & 0 & 0 & 0 & 0 \\
		0 & 0 & 0 & 0 & 0 & 0 & 1 & 1 & 1 & 0 & 0 & 0 & 0 & 0 & 0 & 0 & 1 & 0 & 0 & 0 \\
		1 & 0 & 0 & 0 & 0 & 0 & 0 & 0 & 0 & 1 & 1 & 0 & 0 & 0 & 0 & 0 & 0 & 1 & 1 & 0 \\
		0 & 1 & 0 & 0 & 0 & 0 & 0 & 0 & 0 & 0 & 0 & 1 & 1 & 0 & 0 & 0 & 0 & 1 & 0 & 1 \\
		0 & 0 & 1 & 0 & 0 & 0 & 0 & 0 & 0 & 0 & 0 & 0 & 0 & 1 & 1 & 0 & 0 & 0 & 1 & 1
	\end{array}
	\right].\label{eq:defD}
\end{align}
\hrulefill\par
}%
\end{figure*}
Let us  first establish that $\mathfrak{D}_{\achi} \subseteq \mathfrak{D}^*$. 
This is argued as follows. $(\lambda_1,\lambda_2, \lambda_3, \cdots, \lambda_{20})$ are the amortized amounts of usage of the corresponding building block protocols (\textbf{P1}--\textbf{P20}). Since the batch size $L$ can be chosen to be any positive integer, $\lambda_i$ are allowed to be any non-negative rationals. Therefore, as long as there exist such non-negative $\lambda_{[20]}$ that satisfy the condition in \eqref{eq:region_achievable}, a feasible coding scheme can be constructed from the combination of the aforementioned building block protocols. Denote $\overline{\mathfrak{D}_{\achi}}$ as the closure of $\mathfrak{D}_{\achi}$ in $\mathbb{R}^3$. It then follows that  $\overline{\mathfrak{D}_{\achi}} \subseteq \mathfrak{D}^*$ as $\mathfrak{D}^*$ is closed by definition. To obtain $\overline{\mathfrak{D}_{\achi}}$, let
{\small
\begin{align}
	\mathfrak{D} \triangleq \left\{ (\Delta_{[3]}, \lambda_{[20]}) \in \mathbb{R}^{23}  \mid \mbox{\textbf{Cond1}}, \mbox{\textbf{Cond2}}(\mathbb{Q}) \right\},
\end{align}
}%
where the conditions $\mbox{\textbf{Cond1}}$ and $\mbox{\textbf{Cond2}}(\mathbb{Q})$ have appeared in \eqref{eq:region_achievable}. It is readily seen that the closure of $\mathfrak{D}$ in $\mathbb{R}^{23}$ is equal to
{\small
\begin{align}
	\overline{\mathfrak{D}} = \left\{ (\Delta_{[3]}, \lambda_{[20]}) \in \mathbb{R}^{23}  \mid \mbox{\textbf{Cond1}}, \mbox{\textbf{Cond2}}(\mathbb{R}) \right\}
\end{align}
}%
where $\mbox{\textbf{Cond2}}(\mathbb{R})$ is the condition $\mbox{\textbf{Cond2}}(\mathbb{Q})$ with $\mathbb{Q}$ replaced by $\mathbb{R}$.

Let $\phi$ be the mapping from $\mathbb{R}^{23}$ to $\mathbb{R}^3$ such that
{\small
\begin{align}
	\phi(\Delta_{[3]}, \lambda_{[20]}) = \Delta_{[3]}.
\end{align}
}%
It follows that $\mathfrak{D}_{\achi} = \phi(\mathfrak{D})$. Since $\phi$ is continuous, $ \phi(\overline{\mathfrak{D}}) \subseteq \overline{\mathfrak{D}_{\achi}}$ \cite[Ex. 9.7]{sutherland2009introduction}. On the other hand, $\mathfrak{D}_{\achi} \subseteq \phi(\overline{\mathfrak{D}})$, and thus $\overline{\mathfrak{D}_{\achi}} \subseteq \overline{\phi(\overline{\mathfrak{D}})} = \phi(\overline{\mathfrak{D}})$, where the last step is because $\phi(\overline{\mathfrak{D}})$ is a 3-dimensional polyhedron and thus closed. We conclude that 
{\small
\begin{align}
	\overline{\mathfrak{D}_{\achi}} &= \phi(\overline{\mathfrak{D}}) = \left\{  \Delta_{[3]}  \in \mathbb{R}^{3}  \mid \mbox{\textbf{Cond1}}, \mbox{\textbf{Cond2}}(\mathbb{R}) \right\}.
\end{align}
}%

\subsection{Eliminating the auxiliaries}
Recall that our goal  is to show that $(\mbox{RHS of } \eqref{eq:region_standard}) \subseteq \mathfrak{D}^*$. Since $\overline{\mathfrak{D}_{\achi}} \subseteq \mathfrak{D}^*$, it suffices to show that $(\mbox{RHS of } \eqref{eq:region_standard}) \subseteq \overline{\mathfrak{D}_{\achi}}$. This is done by Fourier-Motzkin elimination. We also show this explicitly in Appendix \ref{proof:elimination}.

\section{Conclusion}
The information theoretic optimality of the $N$-sum box protocol for all $3$ transmitters $\mathbb{F}_q$ linear computations in the LC-QMAC setting, as established in this work, is both promising and intriguing. In particular, it motivates a natural follow up question -- does this optimality hold for any number of transmitters? Since the $N$-sum box is constrained primarily by the SSO condition, generalized optimality results could shed light on the information theoretic significance of this condition. Based on this work, one would expect that generalizations beyond $3$ transmitters might require both new converse bounds, as well as larger subspace-decompositions. In view of the significant challenges associated with establishing, verifying, and preserving high-fidelity multipartite entanglement, generalizations of LC-QMAC that explicitly address resource costs, scalability, and the impact of noise are essential for connecting the theoretical framework to practical real-world systems. Along this direction, the previously established optimality of $N$-sum box protocols in the $\Sigma$-QMAC \cite{Yao_Jafar_Sum_MAC} under generalized entanglement distribution patterns, and in the $\Sigma$-QEMAC \cite{Yao_Jafar_SQEMAC} (where the quantum channels are subject to erasures) bodes well for future efforts towards these generalizations.

\appendix
\section{Eliminating auxiliaries} \label{proof:elimination}
Our goal is to show that, given non-negative $(\Delta_1, \Delta_2, \Delta_3)$ and $(n_{123}, n_{12},\cdots, n_{3})$ that satisfy
{\small
\begin{align} \label{eq:constraint_given}
	{\bm A} 
	\begin{bmatrix}
		\Delta_1\\\Delta_2\\\Delta_3
	\end{bmatrix}
	\geq
	{\bm B} 
	\begin{bmatrix}
		n_{123}\\
		n_{12}\\
		n_{13}\\
		n_{23}\\
		n_o\\
		n_1\\
		n_2\\
		n_3
	\end{bmatrix},
\end{align}
}%
where ${\bm A}$ and ${\bm B}$ are defined in \eqref{eq:def_AB}, there exist non-negative $(\lambda_1,\cdots, \lambda_{20})$ such that
{\small
\begin{align} \label{eq:constraint_delta_lambda}
	\begin{bmatrix}
			\Delta_1\\\Delta_2\\\Delta_3
		\end{bmatrix}
		\geq
		{\bm C} 
		\begin{bmatrix}
			\lambda_1\\ \lambda_2 \\ \lambda_3 \\ \vdots \\ \lambda_{20}
		\end{bmatrix},
\end{align}
}%
{\small
\begin{align} \label{eq:constraint_lambda_n}
	{\bm D}
		\begin{bmatrix}
			\lambda_1\\ \lambda_2 \\ \lambda_3 \\ \vdots \\ \lambda_{20}
		\end{bmatrix}
		\geq
		\begin{bmatrix}
			n_{123}\\ n_{12} \\ n_{13} \\ n_{23} \\ n_o \\ n_1 \\ n_2 \\ n_3
		\end{bmatrix}.
\end{align}
}%
where ${\bm C}$ and ${\bm D}$ are defined in \eqref{eq:defD}.
Let us use analogy for intuition. First note that all variables considered are non-negative reals. Let $\Delta_1, \Delta_2, \Delta_3$ be the  amounts of three non-exchangeable currencies, namely Currency$_1$, Currency$_2$ and Currency$_3$, say corresponding to $3$ different countries, that are available to an importer of goods, subject to the constraint \eqref{eq:constraint_given}. Let \textbf{P1}--\textbf{P20} represent $20$ different goods,  and $\lambda_1, \lambda_2, \cdots, \lambda_{20}$ be the amounts of these goods to be imported, respectively.  ${\bm C}$ specifies the prices of the $20$ goods sold by the three countries. Specifically, the $(i,j)^{th}$ entry of ${\bm C}$, i.e., $C_{i,j}$ is the cost in terms of  Currency$_i$ to import a unit of \textbf{P}$j$. Condition \eqref{eq:constraint_delta_lambda} says that the total amount of any type of currency spent cannot exceed the amount of that type of currency given to the importer. Further constraints are specified by ${\bm D}$: each row in \eqref{eq:constraint_lambda_n} places a demand on the amounts of goods that need to be imported. There are $8$ rows in ${\bm D}$. Let us refer to the $8$ requirements as \textbf{R1} -- \textbf{R8}. For example, the first row of \eqref{eq:constraint_lambda_n} corresponds to \textbf{R1}, and with ${\bm D}$ as defined in \eqref{eq:defD}, this constraint says that the total amount of \textbf{P16} to \textbf{P20} imported has to be at least $n_{123}$. We will show that as long as the importer is given the amount of currencies $(\Delta_1, \Delta_2, \Delta_3)$ that satisfy \eqref{eq:constraint_given}, then there is always a strategy (described as follows) to satisfy all the constraints.

The strategy is divided into the following main steps.
\begin{enumerate}
	\item Satisfy \textbf{R2}, \textbf{R3}, \textbf{R4} by importing $n_{12}$ units of \textbf{P4}, $n_{13}$ units of \textbf{P5} and $n_{23}$ units of \textbf{P6}. This incurs a cost $0.5(n_{12}+n_{13}, n_{12}+n_{23}, n_{13}+n_{23})$ in terms of (Currency$_1$, Currency$_2$, Currency$_3$), and the feasibility (availability of sufficient currency) is guaranteed by \eqref{eq:constraint_given}.
	\item Import $\min\{n_{123}, n_o\} \triangleq \tilde{n}$ unit of \textbf{P17}, which incurs a cost $(\tilde{n}, \tilde{n}, \tilde{n})$. The feasibility is guaranteed by \eqref{eq:constraint_given}. Note that after this step, either \textbf{R1} or \textbf{R5} is satisfied: if $\tilde{n} = n_{123}$, then \textbf{R1} is satisfied; if $\tilde{n} = n_{o}$, then \textbf{R5} is satisfied.
	\item Case I: If $\tilde{n} = n_{123}$, then import appropriate amount of \textbf{P7}, \textbf{P8}, \textbf{P9} to satisfy \textbf{R5}, and import appropriate amount of \textbf{P1}--\textbf{P3}, \textbf{P10}--\textbf{P15} to satisfy \textbf{R6}--\textbf{R8}. Case II: if $\tilde{n} = n_{o}< n_{123}$, then import appropriate amount of \textbf{P1}--\textbf{P3}, \textbf{P10}--\textbf{P16}, \textbf{P18}--\textbf{P20} to satisfy \textbf{R1} and \textbf{R6}--\textbf{R8}.
\end{enumerate}
While in the first two steps we specify exact amount of imported goods and the currencies spent, the third step is more complicated and needs further analysis, since it involves further optimizations that are not so straightforward. In the following we analyze the third step.

Recall that the initial currency amounts given to the importer are $(\Delta_1, \Delta_2, \Delta_3)$. Thus, after the first two steps, there remains
{\small
\begin{align}
	&(\Delta_1, \Delta_2, \Delta_3) -\frac{(n_{12}+n_{13}, n_{12}+n_{23}, n_{13}+n_{23})}{2}  -(\tilde{n}, \tilde{n}, \tilde{n}) \notag \\
	& \triangleq (\Delta_1', \Delta_2', \Delta_3')
\end{align}
}%
currency for the importer to allocate.

\noindent \textbf{Case I:} In this case $\tilde{n} = n_{123}$. Define $n_o' \triangleq n_o-n_{123} \geq 0$. After the first two steps, the importer still needs to fulfill \textbf{R5}--\textbf{R8}. For \textbf{R5}, since $\tilde{n} = n_{123}$ out of $n_o$ is satisfied by importing \textbf{P17}, there remains another $n_o'$ to be fulfilled. The first four rows of  \eqref{eq:constraint_given}  imply,
{\small
\begin{align}
\begin{cases}
	\Delta_1' \geq \frac{1}{2}(n_o'+n_1) \\
	\Delta_2' \geq \frac{1}{2}(n_o'+n_2) \\
	\Delta_3' \geq \frac{1}{2}(n_o'+n_3) \\
	\Delta_1' + \Delta_2' + \Delta_3' \geq 2n_o' + n_1 + n_2 + n_3
\end{cases}
\end{align}
}%
Therefore, there exist non-negative $(a_i, b_i)_{i\in[3]}$ such that
{\small
\begin{align}
	\Delta_i' = \frac{1}{2}(n_o'+n_i) + a_i+b_i, ~~ \forall i\in [3],
\end{align}
}%
and
{\small
\begin{align}
	a_1+a_2+a_3 = \frac{n_o'}{2},~~ b_1+b_2+b_3 = \frac{n_1+n_2+n_3}{2}.
\end{align}
}%
The strategy then finds
{\small
\begin{align}
	&\begin{bmatrix}
		1 & 0.5 & 0.5\\
		0.5 & 1 & 0.5\\
		0.5 & 0.5 & 1 
	\end{bmatrix}
	\begin{bmatrix}
		\lambda_7\\ \lambda_8 \\ \lambda_9
	\end{bmatrix}
	=
	\begin{bmatrix}
		\frac{n_o'}{2}+a_1\\
		\frac{n_o'}{2}+a_2\\
		\frac{n_o'}{2}+a_3\\
	\end{bmatrix}
	\\
	& 
	\implies
	\begin{bmatrix}
		\lambda_7 \\ \lambda_8 \\ \lambda_9
	\end{bmatrix}
	=
	\begin{bmatrix}
		\frac{n_o'}{4}+\frac{3a_1-a_2-a_3}{2}\\
		\frac{n_o'}{4}+\frac{3a_2-a_1-a_3}{2}\\
		\frac{n_o'}{4}+\frac{3a_3-a_1-a_2}{2}\\
	\end{bmatrix} \notag \\
	&~~~~~~~~~~~~~~~~~~~~\geq
	\begin{bmatrix}
		\frac{n_o'}{4}-\frac{a_1+a_2+a_3}{2}\\
		\frac{n_o'}{4}-\frac{a_1+a_2+a_3}{2}\\
		\frac{n_o'}{4}-\frac{a_1+a_2+a_3}{2}\\
	\end{bmatrix}
	=
	\begin{bmatrix}
		0\\0\\0
	\end{bmatrix}
\end{align}
}%
With this choice of $(\lambda_7, \lambda_8, \lambda_9)$, the remaining $n_o'$ part in \textbf{R5} is satisfied, as $\lambda_7+\lambda_8+\lambda_9 = n_o'$. Now, only \textbf{R6} -- \textbf{R8} remain to be fulfilled, and the remaining currency amounts are,
{\small
\begin{align}
 \left(\frac{n_1}{2} + b_1, \frac{n_2}{2} + b_2, \frac{n_3}{2} + b_3\right) \triangleq (\Delta_1'', \Delta_2'', \Delta_3''). 
\end{align}
}%
The importer will then import \textbf{P1}--\textbf{P3}, \textbf{P10}--\textbf{P15} to satisfy \textbf{R6}--\textbf{R8}. We claim that this is feasible as long as the following  conditions hold,
{\small
\begin{align} \label{eq:cond_indep_messages}
\begin{cases}
	\Delta_i'' \geq \frac{n_i}{2}, \forall i \in [3],\\
	\Delta_1''+\Delta_2''+\Delta_3'' \geq n_1+n_2+n_3
\end{cases}
\end{align}
}%
Intuitively, this claim (formalized in Lemma \ref{lem:IM}) follows from the fact that \textbf{P1}--\textbf{P3} are from TQC and \textbf{P10}--\textbf{P15} are from superdense coding. This condition \eqref{eq:cond_indep_messages} is satisfied because $b_1,b_2,b_3$ are non-negative and because $b_1+b_2+b_3= \frac{n_1+n_2+n_3}{2}$. Therefore, \textbf{R5}--\textbf{R8} are satisfied. This completes the proof for Case I.

The claim is formalized in the following lemma, which will be useful again in the sequel.
\begin{lemma} \label{lem:IM}
	Say the remaining demands to be satisfied are $n_1,n_2,n_3$ corresponding to \textbf{R6}, \textbf{R7}, \textbf{R8}, respectively. If the remaining currencies $(\Delta_1, \Delta_2, \Delta_3)$ satisfy $\Delta_i \geq \frac{n_i}{2}, \forall i\in [3]$ and $\Delta_1+\Delta_2+\Delta_3 \geq n_1+n_2+n_3$, then there exist non-negative $(\lambda_i)_{i\in \{1,2,3,10,\cdots,15\}}$ such that
	{\small
	\begin{align}
	\begin{cases}
		\Delta_1 \geq \lambda_1 + 0.5(\lambda_{10} + \lambda_{11} + \lambda_{12} + \lambda_{14}) \\
		\Delta_2 \geq \lambda_2 + 0.5(\lambda_{10} + \lambda_{12} + \lambda_{13} + \lambda_{15}) \\
		\Delta_3 \geq \lambda_3 + 0.5(\lambda_{11} + \lambda_{13} + \lambda_{14} + \lambda_{15}) \\
		\lambda_1 + \lambda_{10}+\lambda_{11} \geq n_{1}\\
		\lambda_2 + \lambda_{12}+ \lambda_{13} \geq n_{2}\\
		\lambda_3 + \lambda_{14}+ \lambda_{15} \geq n_{3}
	\end{cases}\label{eq:lemIM}
	\end{align}
	}%
\end{lemma}
Note that the first three conditions in \eqref{eq:lemIM} imply that the available currency is sufficient for the amounts corresponding to $(\lambda_i)_{i\in \{1,2,3,10,\cdots,15\}}$, with the remaining $\lambda_i$ set to zero. The last three conditions in \eqref{eq:lemIM} imply that \textbf{R6}, \textbf{R7}, \textbf{R8} are satisfied.
\begin{proof}
	Let us first show the existence of $(\lambda_i)_{i\in \{1,2,3,10,\cdots,15\}}$ given that $(\Delta_1, \Delta_2, \Delta_3)$ is in the following set of three extremal points, 
	{\small
	\begin{align*}
		&\mathcal{S} \triangleq \Bigg\{ \left(\frac{n_1}{2},\frac{n_2}{2},\frac{n_1+n_2}{2}+n_3\right), \left(\frac{n_1}{2},\frac{n_1+n_3}{2}+n_2,\frac{n_3}{2}\right),  \left(\frac{n_2+n_3}{2}+n_1,\frac{n_2}{2}, \frac{n_3}{2}\right) \Bigg\}.
	\end{align*}
	}%
	Due to symmetry it suffices to consider the first case, i.e., $(\Delta_1, \Delta_2, \Delta_3) = (\frac{n_1}{2},\frac{n_2}{2},\frac{n_1+n_2}{2}+n_3)$. The solution of $(\lambda_i)_{i\in \{1,2,3,10,\cdots,15\}}$ for this point is listed explicitly as,
	{\small
	\begin{align}
		&\lambda_3 = n_3, \lambda_{11} = n_1, \lambda_{13} = n_2, \notag \\ &\mbox{and} ~\lambda_i = 0, \forall i \in\{1,2,10,12,14,15\}.
	\end{align}
	}%
	By symmetry this shows that the three points in $\mathcal{S}$ are contained in the region specified by \eqref{eq:lemIM}. Let $\mathcal{S}_c$ be the set of all convex combinations of the three points in $\mathcal{S}$. Since the region specified by \eqref{eq:lemIM} is convex, it also contains $\mathcal{S}_c$. This proves the existence of $(\lambda_i)_{i\in \{1,2,3,10,\cdots,15\}}$ for any $(\Delta_1,\Delta_2,\Delta_3) \in  \mathcal{S}_c$. Finally, let us note that all cases of $(\Delta_1, \Delta_2, \Delta_3)$ are either in $\mathcal{S}_c$ or allow more available currency amounts which cannot hurt the existence of $(\lambda_i)_{i\in \{1,2,3,10,\cdots,15\}}$.  
	The proof of the lemma is thus complete.
\end{proof}
\noindent \textbf{Case II:} In this case $\tilde{n} = n_{o}$ and $n_{123}>n_{o}$. Define $n_{123}' \triangleq n_{123}-n_o >0$. After the first two steps, the importer still needs to fulfill \textbf{R1}, \textbf{R6}--\textbf{R8}. For \textbf{R1}, since $\tilde{n} = n_{o}$ out of the $n_{123}$ constraint is already satisfied by importing \textbf{P17}, there only remains another $n_{123}'$ to be fulfilled. To this end, we will show that it suffices  to import certain amounts of \textbf{P1}--\textbf{P3}, \textbf{P10}--\textbf{P16}, \textbf{P18}--\textbf{P20}. Starting with \eqref{eq:constraint_given}, we note that the remaining currency amounts $(\Delta_1', \Delta_2', \Delta_3')$ after the first two steps satisfy
{\small
\begin{align} \label{eq:constraint_given_reduced}
	\begin{cases}
		\Delta_1' \geq \frac{1}{2}(n_{123}' + n_1)\\
		\Delta_2' \geq \frac{1}{2}(n_{123}' + n_2)\\
		\Delta_3' \geq \frac{1}{2}(n_{123}' + n_3)\\
		\Delta_1'+\Delta_2'+\Delta_3' \geq n_{123}' + \frac{n_1+n_2+n_3}{2} + \Gamma \\
		\Delta_1'+\frac{\Delta_2'}{2} + \frac{\Delta_3'}{2} \geq n_{123}'+\frac{n_1+n_2+n_3}{2}   \\
		\Delta_2'+\frac{\Delta_1'}{2} + \frac{\Delta_3'}{2} \geq n_{123}'+\frac{n_1+n_2+n_3}{2}  \\
		\Delta_3'+\frac{\Delta_1'}{2} + \frac{\Delta_2'}{2} \geq n_{123}'+\frac{n_1+n_2+n_3}{2} 
	\end{cases}
\end{align}
}%
where we define,
{\small
\begin{align}
	\Gamma \triangleq \max \Big\{\frac{n_{123}'+n_1}{2},  \frac{n_{123}'+n_2}{2},\frac{n_{123}'+n_3}{2}, \frac{n_1+n_2+n_3}{2} \Big\}.
\end{align}
}%
{\small
\begin{align}
	\begin{cases}
		\Delta_1' \geq \lambda_1 + \frac{\lambda_{10}+\lambda_{11}+\lambda_{12}+\lambda_{14}+\lambda_{16}}{2}+\lambda_{18}+\lambda_{19}+\lambda_{20}\\
		\Delta_2' \geq \lambda_2 + \frac{\lambda_{10}+\lambda_{12}+\lambda_{13}+\lambda_{15}+\lambda_{16}}{2}+\lambda_{18}+\lambda_{19}+\lambda_{20} \\
		\Delta_3' \geq \lambda_3 + \frac{\lambda_{11}+\lambda_{13}+\lambda_{14}+\lambda_{15}+\lambda_{16}}{2}+\lambda_{18}+\lambda_{19}+\lambda_{20} \\
		\lambda_{16}+\lambda_{18} + \lambda_{19}+\lambda_{20} \geq n_{123}' ~~ \mbox{\textbf{(R1)}} \\
		\lambda_1 + \lambda_{10} + \lambda_{11} + \lambda_{18} + \lambda_{19} \geq n_1 ~~ \mbox{\textbf{(R6)}} \\
		\lambda_2 + \lambda_{12} + \lambda_{13} + \lambda_{18} + \lambda_{20} \geq n_2 ~~ \mbox{\textbf{(R7)}} \\
		\lambda_3 + \lambda_{14} + \lambda_{15} + \lambda_{19} + \lambda_{20} \geq n_3 ~~ \mbox{\textbf{(R8)}}
	\end{cases}
\end{align}
}%
It suffices to show the existence of $(\lambda_i)_{i\in\{1,2,3,10,\cdots,16, 18,19,20\}}$ for the corner points of $(\Delta_1', \Delta_2', \Delta_3')$ in the region induced by \eqref{eq:constraint_given_reduced}. Further by symmetry, it suffices to consider the following $7$ subcases (II.1 -- II.7).

\noindent \textbf{II.1:} In this case we consider
{\small
\begin{align}
	\begin{bmatrix}
		\Delta_1' \\ \Delta_2' \\ \Delta_3'
	\end{bmatrix}
	=
	\begin{bmatrix}
		\frac{n_{123}'+n_1}{2} \\
		\frac{n_{123}'+n_2}{2} \\
		\frac{n_{123}'+n_3}{2}
	\end{bmatrix}
\end{align}
}%
which corresponds to the first three inequalities in \eqref{eq:constraint_given_reduced} being tight.
It can be verified that \eqref{eq:constraint_given_reduced} then implies
{\small
\begin{align}
	n_1=n_2=n_3=0,
\end{align}
}%
by noting the non-negativity of $n_i, \forall i\in [3]$. This means that \textbf{R5}--\textbf{R8} do not require anything.
Therefore, for this subcase, the importer only needs to import $n_{123}'$ amount of \textbf{P16} to satisfy \textbf{R1}. The feasibility is guaranteed since $\Delta_i' \geq \frac{n_{123}'}{2}$ for $i\in [3]$.

\noindent \textbf{II.2:} 
In this case we consider
{\small
\begin{align}
	\begin{bmatrix}
		\Delta_1' \\ \Delta_2' \\ \Delta_3'
	\end{bmatrix}
	=
	\begin{bmatrix}
		\frac{n_{123}'+n_1}{2} \\
		\frac{n_{123}'+n_2}{2} \\
		\frac{n_3}{2} + \Gamma
	\end{bmatrix}
\end{align}
}%
which corresponds to the $1^{st}, 2^{nd}, 4^{th}$ inequalities in \eqref{eq:constraint_given_reduced} being tight. It can be verified that \eqref{eq:constraint_given_reduced} then implies
{\small
\begin{align} \label{eq:II2}
\begin{cases}
	n_1 \geq \min\{n_2+n_3, n_{123}'\} \\
	n_2 \geq \min\{n_1+n_3, n_{123}'\}
\end{cases},
\end{align}
}%
and we consider the following subcases.

\noindent \textbf{II.2.a:} $n_{123}' \geq \max\{n_1+n_3, n_2+n_3\}$. \eqref{eq:II2} implies $n_1=n_2, n_3= 0$. Import $n_1$ amount of \textbf{P18}, and $n_{123}'-n_1$ amount of \textbf{P16}. This is feasible as $\Delta_i'\geq \frac{n_{123}'+n_1}{2}, \forall i\in [3]$.

\noindent \textbf{II.2.b:} $n_1+n_3 \geq n_{123}' \geq  n_2+n_3$. \eqref{eq:II2} implies $n_1\geq n_2+n_3, n_2 \geq n_{123}'$. This further implies $n_3 = 0$ (\textbf{R8} satisfied) and $n_{123}' = n_2$. Import $n_2$ \textbf{P18}, and then \textbf{R1}, \textbf{R7} are satisfied. It remains to satisfy \textbf{R6}. The remaining currencies are
{\small
\begin{align}
	\begin{bmatrix}
		\Delta_1'' \\ \Delta_2'' \\ \Delta_3''
	\end{bmatrix}
	=
	\begin{bmatrix}
		\frac{n_1-n_{123}'}{2}  \\ 0 \\ \Gamma - n_{123}'
	\end{bmatrix}
\end{align}
}%
and the remaining demands for \textbf{R6} is $n_1 - n_{123}'$. Lemma \ref{lem:IM} then implies that the remaining demand  of \textbf{R6} can be satisfied with the remaining currencies.

\noindent \textbf{II.2.c:} $n_1+n_3 \geq  n_2+n_3 \geq n_{123}'$. \eqref{eq:II2} implies $n_1\geq n_{123}', n_2 \geq n_{123}'$. Import $n_{123}'$ \textbf{P18} and \textbf{R1} is satisfied. The remaining currency amounts are
{\small
\begin{align}
	\begin{bmatrix}
		\Delta_1'' \\ \Delta_2'' \\ \Delta_3''
	\end{bmatrix}
	=
	\begin{bmatrix}
		\frac{n_1-n_{123}'}{2} \\ \frac{n_2-n_{123}'}{2} \\ \frac{n_3}{2} + \Gamma - n_{123}'
	\end{bmatrix}
\end{align}
}%
and the remaining demands for \textbf{R6}--\textbf{R8} are
{\small
\begin{align}
\begin{bmatrix}
	n_1' \\ n_2' \\ n_3'
\end{bmatrix}
=
\begin{bmatrix}
	n_1 - n_{123}'\\
	n_2 - n_{123}'\\
	n_3
\end{bmatrix}.
\end{align}
}%
Lemma \ref{lem:IM} then implies that the remaining demands of \textbf{R6}--\textbf{R8} can be satisfied with the remaining currencies.

\noindent \textbf{II.2.d:} $n_2+n_3 \geq n_{123}' \geq  n_1+n_3$. By symmetry this case can be reduced to \textbf{II.2.b}.

\noindent \textbf{II.2.e:} $n_2+n_3 \geq  n_1+n_3 \geq n_{123}'$. By symmetry this case can be reduced to \textbf{II.2.c}.

\noindent \textbf{II.3:} In this case we consider
{\small
\begin{align}
\begin{bmatrix}
		\Delta_1' \\ \Delta_2' \\ \Delta_3'
	\end{bmatrix}
	=
	\begin{bmatrix}
		\frac{n_{123}'+n_1}{2} \\
		\frac{n_{123}'+n_2}{2} \\
		\frac{n_{123}'}{2}+\frac{n_2}{2}+n_3
	\end{bmatrix}
\end{align}
}%
which corresponds to the $1^{st}, 2^{nd}, 5^{th}$ inequalities in \eqref{eq:constraint_given_reduced} being tight. It can be verified that \eqref{eq:constraint_given_reduced} then implies
{\small
\begin{align}
	n_1\leq \min \{n_{123}', n_2\}.
\end{align}
}%
Import $n_1$ amount of \textbf{P18} and $n_{123}'-n_1$ amount of \textbf{P16}. \textbf{R1} is then satisfied. The remaining currency amounts are
{\small
\begin{align}
	\begin{bmatrix}
		\Delta_1'' \\ \Delta_2'' \\ \Delta_3''
	\end{bmatrix}
	=
	\begin{bmatrix}
		0 \\ \frac{n_2-n_1}{2} \\ \frac{n_2-n_1}{2}+n_3
	\end{bmatrix}
\end{align}
}%
and the remaining demands for \textbf{R6}--\textbf{R8} are
{\small
\begin{align}
\begin{bmatrix}
	n_1' \\ n_2' \\ n_3'
\end{bmatrix}
=
\begin{bmatrix}
	0 \\
	n_2 - n_1\\
	n_3
\end{bmatrix}.
\end{align}
}%
Lemma \ref{lem:IM} then implies that the remaining demands of \textbf{R6}--\textbf{R8} can be satisfied with the remaining currencies. 

\noindent \textbf{II.4:} In this case we consider
{\small
\begin{align}
\begin{bmatrix}
		\Delta_1' \\ \Delta_2' \\ \Delta_3'
	\end{bmatrix}
	=
	\begin{bmatrix}
		\frac{n_{123}'+n_1}{2} \\
		\frac{n_{123}'+n_2}{2} \\
		\frac{n_{123}'}{2}+\frac{n_1+n_2}{4}+\frac{n_3}{2}
	\end{bmatrix}
\end{align}
}%
which corresponds to the $1^{st}, 2^{nd}, 7^{th}$ inequalities in \eqref{eq:constraint_given_reduced} being tight. It can be verified that \eqref{eq:constraint_given_reduced} then implies
{\small
\begin{align}
	n_{123}' \geq n_1=n_2, ~ n_3= 0.
\end{align}
}%
\textbf{R8} requires nothing. Import $n_1$ amount of \textbf{P18}, and $n_{123}'-n_1$ amount of \textbf{P16}. The feasibility can be verified and this satisfies \textbf{R1}, \textbf{R6} and \textbf{R7}. 

\noindent \textbf{II.5:} 
In this case we consider
{\small
\begin{align}
\begin{bmatrix}
		\Delta_1' \\ \Delta_2' \\ \Delta_3'
	\end{bmatrix}
	=
	\begin{bmatrix}
		\frac{n_{123}'+n_1}{2}\\
		\frac{n_1+n_2+n_3}{2}+n_{123}'-\Gamma\\
		2\Gamma - \frac{n_1}{2} - \frac{n_{123}'}{2}
	\end{bmatrix}
\end{align}
}%
which corresponds to the $1^{st}, 4^{th}, 6^{th}$ inequalities in \eqref{eq:constraint_given_reduced} being tight. Note that $\Gamma$ is a maximum of $4$ terms, and we further consider subcases according to the value of $\Gamma$ as follows.

\noindent \textbf{II.5.a:} $\Gamma = \frac{n_{123}'+n_1}{2}$. This condition together with \eqref{eq:constraint_given_reduced} implies
{\small
\begin{align}
	\min\{n_1,n_{123}'\} \geq n_2+n_3.
\end{align}
}%
Import $n_2$ amount of \textbf{P18}, $n_3$ amount of \textbf{P19}, and $n_{123}'-(n_2+n_3)$ amount of \textbf{P16}. \textbf{R1}, \textbf{R7} and \textbf{R8} are then satisfied. 
The remaining currencies are
{\small
\begin{align}
	\begin{bmatrix}
		\Delta_1'' \\ \Delta_2'' \\ \Delta_3''
	\end{bmatrix}
	=
	\begin{bmatrix}
		\frac{n_1-n_2-n_3}{2} \\
		0 \\
		\frac{n_1-n_2-n_3}{2}
	\end{bmatrix}
\end{align}
}%
and the remaining demand for \textbf{R6} is
{\small
\begin{align}
	n_3' = n_1 - n_2 - n_3.
\end{align}
}%
Lemma \ref{lem:IM} then implies that the remaining demand of \textbf{R8} can be satisfied with the remaining currency amounts.

\noindent \textbf{II.5.b:} $\Gamma = \frac{n_{123}'+n_2}{2}$. This condition together with \eqref{eq:constraint_given_reduced} implies
{\small
\begin{align}
	n_3=0,~~ n_{123}'\geq n_1=n_2.
\end{align}
}%
\textbf{R8} requires nothing. Import $n_2$ amount of \textbf{P18} and $n_{123}'-n_2$ amount of \textbf{P16}. The feasibility can be verified, and this satisfies \textbf{R1}, \textbf{R6} and \textbf{R7}.

\noindent \textbf{II.5.c:} $\Gamma = \frac{n_{123}'+n_3}{2}$. This condition together with \eqref{eq:constraint_given_reduced} implies
{\small
\begin{align}
	n_2 = 0,~~ \min\{n_{123}', n_3\} \geq n_1.
\end{align}
}%
\textbf{R7} requires nothing. Import $n_1$ amount of \textbf{P19}, $n_{123}'-n_1$ amount of \textbf{P16}, and $n_3-n_1$ amount of \textbf{P3}. The feasibility can be verified and  this satisfies \textbf{R1},  \textbf{R6} and \textbf{R8}.

\noindent \textbf{II.5.d:} $\Gamma = \frac{n_1+n_2+n_3}{2}$. This condition together with \eqref{eq:constraint_given_reduced} implies
{\small
\begin{align}
	n_1\geq n_{123}'\geq n_2, n_2+n_3\geq n_{123}'.
\end{align}
}%
Import $n_2$ amount of \textbf{P18}, and $n_{123}'- n_2$ amount of \textbf{P19}. \textbf{R1} and \textbf{R7} are then satisfied. The remaining currencies are
{\small
\begin{align}
	\begin{bmatrix}
		\Delta_1'' \\ \Delta_2'' \\ \Delta_3''
	\end{bmatrix}
	=
	\begin{bmatrix}
		\frac{n_1-n_{123}'}{2} \\
		0\\
		\frac{n_1-3n_{123}'}{2}+n_2+n_3 
	\end{bmatrix}
\end{align}
}%
and the remaining demands for \textbf{R6} and \textbf{R8} are
{\small
\begin{align}
	n_1' = n_1 - n_{123}', ~ n_3' = n_3- n_{123}'+n_2.
\end{align}
}%
Lemma \ref{lem:IM} then implies that the remaining demands of \textbf{R6} and \textbf{R8} can be satisfied with the remaining currencies.

\noindent \textbf{II.6:} In this case we consider
{\small
\begin{align}
	\begin{bmatrix}
		\Delta_1' \\ \Delta_2' \\ \Delta_3'
	\end{bmatrix}
	=
	\begin{bmatrix}
		\frac{n_1+n_2+n_3}{2}+n_{123}'-\Gamma \\
	\frac{n_1+n_2+n_3}{2}+n_{123}'-\Gamma \\
	3\Gamma - \frac{n_1+n_2+n_3}{2} -  n_{123}' 
	\end{bmatrix}
\end{align}
}%
which corresponds to the $4^{th}, 5^{th}, 6^{th}$ inequalities in \eqref{eq:constraint_given_reduced} being tight. We further consider subcases according to the value of $\Gamma$ as follows.

\noindent \textbf{II.6.a:} $\Gamma = \frac{n_{123}'+n_1}{2}$. This condition together with \eqref{eq:constraint_given_reduced} implies
{\small
\begin{align}
	n_{123}' \geq n_1= n_2+n_3.
\end{align}
}%
Import $n_2$ amount of \textbf{P18}, $n_3$ amount of \textbf{P19}, and $n_{123}'-n_1$ amount of \textbf{P16}. The feasibility can be verified, and this satisfies \textbf{R1} and \textbf{R6} -- \textbf{R8}.

\noindent \textbf{II.6.b:} $\Gamma = \frac{n_{123}'+n_2}{2}$. By symmetry, this case is the same as II.6.a.

\noindent \textbf{II.6.c:} $\Gamma = \frac{n_{123}'+n_3}{2}$. This condition together with \eqref{eq:constraint_given_reduced} implies
{\small
\begin{align}
	\min\{n_{123}', n_3 \} \geq n_1+n_2.
\end{align}
}%
Import $n_1$ amount of \textbf{P19}, $n_2$ amount of \textbf{P20}, $n_{123}'-(n_1+n_2)$ amount of \textbf{P16} and $n_3-n_1-n_2$ amount of \textbf{P3}. The feasibility can be verified and this satisfies \textbf{R1} and \textbf{R6} -- \textbf{R8}.

\noindent \textbf{II.6.d:} $\Gamma = \frac{n_1+n_2+n_3}{2}$. This condition together with \eqref{eq:constraint_given_reduced} implies
{\small
\begin{align}
	&\min\Big\{n_1+n_2, n_1+n_3, n_2+n_3, \frac{n_1+n_2+n_3}{2}\Big\} \notag \\
	&\geq n_{123}' \geq \max\{n_1,n_2\}.
\end{align}
}%
Import $n_1+n_2-n_{123}'$ amount of \textbf{P18}, $n_{123}'-n_2$ amount of \textbf{P19}, $n_{123}'-n_1$ amount of \textbf{P20}, and $n_1+n_2+n_3-2n_{123}$ amount of \textbf{P3}.   The feasibility can be verified and this satisfies \textbf{R1} and \textbf{R6}--\textbf{R8}.

\noindent \textbf{II.7:} In this case we consider
{\small
\begin{align}
	\begin{bmatrix}
		\Delta_1' \\ \Delta_2' \\ \Delta_3'
	\end{bmatrix}
	=
	\begin{bmatrix}
		\frac{n_{123}'}{2}+\frac{n_1+n_2+n_3}{4}\\
		\frac{n_{123}'}{2}+\frac{n_1+n_2+n_3}{4}\\
		\frac{n_{123}'}{2}+\frac{n_1+n_2+n_3}{4}
	\end{bmatrix}
\end{align}
}%
which corresponds to the $5^{th}, 6^{th}, 7^{th}$ inequalities in \eqref{eq:constraint_given_reduced} being tight. It can be verified that \eqref{eq:constraint_given_reduced} then implies
{\small
\begin{align}
	&n_i+n_j \geq n_k, \mbox{ for distinct } i,j,k\in[3],\notag \\
	&\mbox{and} ~~ n_{123}'\geq \frac{n_1+n_2+n_3}{2}.
\end{align}
}%
Import $\frac{n_1+n_2-n_3}{2}$ amount of \textbf{P18}, $\frac{n_1+n_3-n_2}{2}$ amount of \textbf{P19}, $\frac{n_2+n_3-n_1}{2}$ amount of \textbf{P20}, and $n_{123}'-\frac{n_1+n_2+n_3}{2}$ amount of \textbf{P16}. The feasibility can be verified and this satisfies \textbf{R1} and \textbf{R6}--\textbf{R8}. \hfill \qed

\section{Necessity of $3$-way Entanglement for Toy Example 1} \label{proof:necessity3}
Recall that the setting in Toy Example 1 contains Alice$_1$, Alice$_2$, Alice$_3$, who have data streams $(A,B), (C,D)$, $(E,F)$, respectively, all symbols in $\mathbb{F}_d$ with $d=3$, and a receiver (Bob) who wishes to compute, 
{\small
\begin{align*}
	f(A,B,C,D,E,F)=\bbsmatrix{A+C+E\\ B+2D\\ B+2F}.
\end{align*}
}%
Suppose instead of all possible quantum coding schemes as specified in the problem formulation, we now only allow the transmitters to use \emph{pairwise} entanglement throughout all the stages. Specifically, Alice$_1$ and Alice$_2$ share a bipartite quantum system $Q_1 = Q_{1,1}Q_{1,2}$ such that $Q_{1,1}$ is accessible at Alice$_1$ and $Q_{1,2}$ is accessible at Alice$_2$. Similarly, Alice$_1$ and Alice$_3$ share another quantum system $Q_2$ such that $Q_{2,k}$ is accessible at Alice$_k$ for $k\in \{1,3\}$; Alice$_2$ and Alice$_3$ share another quantum system $Q_3$ such that $Q_{3,k}$ is accessible at Alice$_k$ for $k\in \{2,3\}$. $Q_1$, $Q_2$ and $Q_3$ are assumed to be independent in the preparation stage, kept unentangled in the encoding stage, and measured separately in the decoding stage, whereas the subsystems $Q_{i,j}$ and $Q_{i,k}$ are allowed to be entangled for distinct $j,k\in \{1,2,3\}$. Let $\delta_i, i\in [3]$ denote the dimension of $Q_i$ in the encoding stage. According to \cite{Yao_Jafar_Sum_MAC}, one can lower bound the total download cost $\sum_{i\in [3]}\log_d \delta_i /L$ by the classical (unentangled) total download cost of a hypothetical problem, where there are $\binom{3}{2}=3$ transmitters, denoted as Alice$_1'$, Alice$_2'$, Alice$_3'$, who know $(A,B,C,D), (A,B,E,F),(C,D,E,F)$, and the same receiver (Bob) who computes the same function $f$.
This is because any output measured from $Q_{i}$ can be sent directly through a same dimension classical system from Alice$_i'$ in the hypothetical setting for $i\in [3]$.
In the hypothetical setting, let $X_i, i\in [3]$ be a  $\delta_i'$-dimensional (classical) system sent from Alice$_i'$. We want to obtain a lower bound for $\sum_{i \in [3]}\log_d \delta_i'/L$.

Without loss of generality, assuming that each of the data streams $A,B,\cdots, F$ is uniformly distributed in $\mathbb{F}_d^L$, we have
{\small
\begin{align}
	& \sum_{i \in [3]}\log_d \delta_i'  \geq H(X_1,X_2,X_3) \\
	& = H(X_1,X_2, X_3, \bbsmatrix{ B+2D \\ B+2F}) \label{eq:nece_1} \\
	& = H(\bbsmatrix{B+2D \\ B+2F}) + H(X_1,X_2,X_3 \mid \bbsmatrix{B+2D \\ B+2F}) \\
	& = 2L + H(X_1,X_2,X_3 \mid \bbsmatrix{B+2D \\ B+2F})\\
	& \geq 2L + \frac{1}{2} \sum_{i=1}^3 H(X_{[3]\setminus \{i\}}   \mid X_i, \bbsmatrix{B+2D \\ B+2F}) \label{eq:nece_2} \\
	& = 2L + \frac{1}{2} \sum_{i=1}^3 H(X_{[3]\setminus \{i\}}, A+C+E   \mid X_i, \bbsmatrix{B+2D \\ B+2F}) \label{eq:nece_3}\\
	& \geq 2L + \frac{1}{2}(3L) \label{eq:nece_4} \\
	& = 3.5L \\
	& \hspace{0cm} \implies \sum_{i\in[3]} \log_d \delta_i'/L \geq 3.5
\end{align}
}%
Step \eqref{eq:nece_1}  holds because $(B+2D, B+2F)$ is determined by $(X_1,X_2,X_3)$.  Step \eqref{eq:nece_2} follows from submodularity of classical entropy, i.e., the general property that $2H(Z_1,Z_2,Z_3\mid Z_4)\geq H(Z_1,Z_2\mid Z_3,Z_4)+H(Z_2,Z_3\mid Z_1,Z_4)+H(Z_3,Z_1\mid Z_2,Z_4)$ for any classical random variables $Z_1,Z_2,Z_3,Z_4$. Step \eqref{eq:nece_3} holds because $A+C+E$ is determined by $(X_1,X_2,X_3)$. To see Step \eqref{eq:nece_4}, note that $(X_1,A,B,C,D,F)$ is independent of $E$, so that the first term in the sum, i.e., $H(X_2,X_3,A+C+E \mid X_1, \bbsmatrix{B+2D \\ B+2F})$ $ \geq H(X_2,X_3,A+C+E\mid A,B,C,D, F, X_1, \bbsmatrix{B+2D \\ B+2F}) \geq H(E) = L$, and similar reasoning applies to each of the three terms in the sum, so that their sum is lower bounded by $3L$. Therefore, the total download cost for the hypothetical problem is at least $3.5$. We conclude that with only $2$-way entanglement, the total download cost for Toy Example $1$ is at least $3.5$.

\end{document}